\documentclass[sigconf, 10pt]{acmart}
\makeatletter
\def\runningfoot{\def\@runningfoot{}}
\def\firstfoot{\def\@firstfoot{}}
\makeatother

\usepackage{graphicx}
\usepackage{subfigure}

\usepackage{balance}  
\usepackage[normalem]{ulem}
\usepackage{booktabs}
\usepackage{listings}
\usepackage{fancyvrb}
\usepackage{braket}
\usepackage{bbold}
\usepackage{bm}
\usepackage{clrscode3e}
\usepackage{colortbl}
\usepackage{forest}
\usepackage{algorithm}
\usepackage{xspace}
\graphicspath{ {imgs/} }
\usepackage{textcomp}
\usepackage{multirow}
\usepackage{booktabs}
\usepackage{color}
\usepackage{tcolorbox}
\usepackage{todonotes}

\renewcommand{\footnotesize}{\small}
\definecolor{black}{rgb}{0,0,0}
\definecolor{grey}{rgb}{0.8,0.8,0.8}
\definecolor{red}{rgb}{1,0,0}
\definecolor{green}{rgb}{0,1,0}
\definecolor{darkgreen}{rgb}{0,0.5,0}
\definecolor{darkpurple}{rgb}{0.5,0,0.5}
\definecolor{darkdarkpurple}{rgb}{0.3,0,0.3}
\definecolor{blue}{rgb}{0,0,1}
\definecolor{shadegreen}{rgb}{0.95,1,0.95}
\definecolor{shadeblue}{rgb}{0.95,0.95,1}
\definecolor{shadered}{rgb}{1,0.85,0.85}
\definecolor{shadegrey}{rgb}{0.85,0.85,0.85}
\definecolor{oddRowGrey}{rgb}{0.80,0.80,0.80}
\definecolor{evenRowGrey}{rgb}{0.85,0.85,0.85}

\newcommand{\red}[1]{{\color{red} #1}}
\newcommand{\cut}[1]{{}}

\newcommand{\mypar}[1]{\smallskip\noindent\textbf{{#1}.}}

\DeclareMathAlphabet{\mathbbold}{U}{bbold}{m}{n}

\newtheorem{Theorem}{Theorem}
\newtheorem{Definition}{Definition}

\newtheorem{Example}{Example}

\newcommand{\NPhard}{NP-hard}
\newcommand{\NPcomplete}{NP-complete}

\newcommand{\theRADB}{the relational algebra interpreter}
\newcommand{\aRADB}{a relational algebra interpreter}
\newcommand{\OurSys}{\textsc{RATest}}

\newcommand{\spjudBasic}{\textsc{Basic}}
\newcommand{\spjudOpt}{\textsc{$Opt_{\sigma}$}}
\newcommand{\aggBasic}{\textsc{Agg-Basic}}
\newcommand{\aggParam}{\textsc{Agg-Param}}
\newcommand{\aggHeu}{\textsc{Agg-Opt}}

\newcommand{\eg}{{e.g.}}
\newcommand{\ie}{{i.e.}}
\newcommand{\constraints}{{\Gamma}}
\newcommand{\witset}{{\mathcal W}}
\newcommand{\SCP}{{\tt SCP}}
\newcommand{\SWP}{{\tt SWP}}
\newcommand{\Prv}{{\tt Prv}}

\newcommand{\selection}{\sigma}
\newcommand{\projection}{\pi}
\newcommand{\join}{\Join}
\newcommand{\union}{\cup}

\newcommand{\aggregation}[2]{\gamma_{{\tt #1, #2}}}  

\renewcommand\footnotetextcopyrightpermission[1]{} 
\begin{document}

\title{Explaining Wrong Queries Using Small Examples}
\author{Zhengjie Miao, Sudeepa Roy, and Jun Yang}
    \affiliation{%
		\institution{Duke University}
    }
\email{{zjmiao,sudeepa,junyang}@cs.duke.edu}
\renewcommand{\shortauthors}{Z. Miao, S. Roy, and J. Yang}

\definecolor{lstpurple}{rgb}{0.5,0,0.5}
\definecolor{lstred}{rgb}{1,0,0}
\definecolor{lstreddark}{rgb}{0.7,0,0}
\definecolor{lstredl}{rgb}{0.64,0.08,0.08}
\definecolor{lstmildblue}{rgb}{0.66,0.72,0.78}
\definecolor{lstblue}{rgb}{0,0,1}
\definecolor{lstmildgreen}{rgb}{0.42,0.53,0.39}
\definecolor{lstgreen}{rgb}{0,0.5,0}
\definecolor{lstorangedark}{rgb}{0.6,0.3,0}	
\definecolor{lstorange}{rgb}{0.75,0.52,0.005}
\definecolor{lstorangelight}{rgb}{0.89,0.81,0.67}
\definecolor{lstbeige}{rgb}{0.90,0.86,0.45}

\DeclareFontShape{OT1}{cmtt}{bx}{n}{<5><6><7><8><9><10><10.95><12><14.4><17.28><20.74><24.88>cmttb10}{}

\lstdefinelanguage{smtlib2}{
  alsoletter=-,
  morekeywords={declare-const,define-fun,assert,minimize,maximize,check-sat,get-objectives,and,or,not,distinct},
  extendedchars=false,
  keywordstyle=\bfseries\color{lstpurple},
  deletekeywords={Int,Bool},
  keywords=[2]{Int,Bool},
  keywordstyle=[2]\color{lstblue},
}

\lstdefinestyle{psql}
{
tabsize=2,
basicstyle=\small\upshape\ttfamily,
language=SQL,
morekeywords={PROVENANCE,BASERELATION,INFLUENCE,COPY,ON,TRANSPROV,TRANSSQL,TRANSXML,CONTRIBUTION,COMPLETE,TRANSITIVE,NONTRANSITIVE,EXPLAIN,SQLTEXT,GRAPH,IS,ANNOT,THIS,XSLT,MAPPROV,cxpath,OF,TRANSACTION,SERIALIZABLE,COMMITTED,INSERT,INTO,WITH,SCN,UPDATED},
extendedchars=false,
keywordstyle=\bfseries,
mathescape=true,
escapechar=@,
sensitive=true
}

\lstdefinestyle{psqlcolor}
{
tabsize=2,
basicstyle=\small\upshape\ttfamily,
language=SQL,
morekeywords={PROVENANCE,BASERELATION,INFLUENCE,COPY,ON,TRANSPROV,TRANSSQL,TRANSXML,CONTRIBUTION,COMPLETE,TRANSITIVE,NONTRANSITIVE,EXPLAIN,SQLTEXT,GRAPH,IS,ANNOT,THIS,XSLT,MAPPROV,cxpath,OF,TRANSACTION,SERIALIZABLE,COMMITTED,INSERT,INTO,WITH,SCN,UPDATED},
extendedchars=false,
keywordstyle=\bfseries\color{lstpurple},
deletekeywords={count,min,max,avg,sum},
keywords=[2]{count,min,max,avg,sum},
keywordstyle=[2]\color{lstblue},
stringstyle=\color{lstreddark},
commentstyle=\color{lstgreen},
mathescape=true,
escapechar=@,
sensitive=true
}

\lstdefinestyle{datalog}
{
basicstyle=\footnotesize\upshape\ttfamily,
language=prolog
}

\lstdefinestyle{pseudocode}
{
  tabsize=3,
  basicstyle=\small,
  language=c,
  morekeywords={if,else,foreach,case,return,in,or},
  extendedchars=true,
  mathescape=true,
  literate={:=}{{$\gets$}}1 {<=}{{$\leq$}}1 {!=}{{$\neq$}}1 {append}{{$\listconcat$}}1 {calP}{{$\cal P$}}{2},
  keywordstyle=\color{lstpurple},
  escapechar=&,
  numbers=left,
  numberstyle=\color{lstgreen}\small\bfseries, 
  stepnumber=1, 
  numbersep=5pt,
}

\lstdefinestyle{xmlstyle}
{
  tabsize=3,
  basicstyle=\small,
  language=xml,
  extendedchars=true,
  mathescape=true,
  escapechar=£,
  tagstyle=\color{keywordpurple},
  usekeywordsintag=true,
  morekeywords={alias,name,id},
  keywordstyle=\color{lstred}
}

\lstdefinestyle{smtlib2}
{
tabsize=2,
basicstyle=\scriptsize\upshape\ttfamily,
numbers=left,
stepnumber=1,
breaklines=true,
stringstyle=\color{lstreddark},
commentstyle=\color{lstgreen},
mathescape=true,
escapechar=@,
sensitive=true
}


\lstset{style=psqlcolor}


\begin{abstract}

  For testing the correctness of SQL queries, e.g., evaluating student
  submissions in a database course, a standard practice is to execute
  the query in question on some test database instance and compare its
  result with that of the correct query.  Given two queries $Q_1$ and
  $Q_2$, we say that a database instance $D$ is a counterexample (for
  $Q_1$ and $Q_2$) if $Q_1(D)$ differs from $Q_2(D)$; such a
  counterexample can serve as an explanation of why $Q_1$ and $Q_2$
  are not equivalent.  While the test database instance may serve as a
  counterexample, it may be too large or complex to read and
  understand where the inequivalence comes from.  Therefore, in this
  paper, given a known counterexample $D$ for $Q_1$ and $Q_2$, we aim
  to find the smallest counterexample $D' \subseteq D$ where
  $Q_1(D') \neq Q_2(D')$.  The problem in general is NP-hard.  We give
  a suite of algorithms for finding the smallest counterexample for
  different classes of queries, some more tractable than others.  We
  also present an efficient provenance-based algorithm for SPJUD
  queries that uses a constraint solver, and extend it to more complex
  queries with aggregation, group-by, and nested queries.  We perform
  extensive experiments indicating the effectiveness and scalability
  of our solution on student queries from an undergraduate database
  course and on queries from the TPC-H benchmark.  We also report a
  user study from the course where we deployed our tool to help
  students with an assignment on relational algebra.
\end{abstract}

\maketitle
\section{Introduction}
\label{sec:introduction}



Correctness of database queries is often validated by evaluating the
queries with respect to a \emph{reference query} and a \emph{reference
  database instance} for testing. A primary application is in teaching
students how to write SQL queries in database courses in academic
institutions and evaluating their solutions. Typically, there is a
test database instance $D$, and a correct query $Q_1$. The correctness
of the query $Q_2$ submitted by a student is validated by checking
whether $Q_1(D) = Q_2(D)$. Assuming that $Q_2$ is at least
syntactically correct and its output schema is compatible with that of
$Q_2$ (which can be easily verified), if $Q_2$ does not solve the
intended problem, then there will be at least one tuple in $Q_1(D)$
and not in $Q_2(D)$, or in $Q_2(D)$ but not in $Q_1(D)$.
Another application scenario is when people rewrite complex SQL
queries to obtain better performance. One approach for checking the
correctness of complex rewritten queries is regression testing:
execute the rewritten query $Q_2$ on test instances $D$ to make sure
that $Q_2$ returns the same results as the original query
$Q_1$. Finding an answer tuple differentiating two queries and
providing an \emph{explanation} for its existence helps students or
developers understand the error and fix their queries.

In both applications above, if the test database $D$ is large---either
because it is a large real data set or it is synthesized to be large
enough to test scalability or ensure coverage of numerous corner
cases---it would take much effort to understand where the
inequivalence of two queries came from. Suppose a database course in a
university uses the DBLP database \cite{dblpdata} in an assignment on
SQL or relational algebra (RA). The DBLP database has more than 5
million entries, and giving this entire database (or the outputs) to
students as a counterexample to their query is not much effective. In
practice, the mistakes in most of the queries can be explained with
only a small number of tuples, which is much more useful as a
counterexample for debugging.

Of course, one could generate a completely different counterexample
$D'$ altogether, but using the test database instance $D$ to help
generate a counterexample has some distinct advantages.  First, it
helps to preserve the same context for users by using the same data
values and relationships.  Second, knowing that the original instance
$D$ is already a counterexample can help create the counterexample
$D'$ more efficiently. This motivates the problem we study in this
paper: \emph{given a reference database $D$, a reference query $Q_1$,
  and a test query $Q_2$ such that $Q_1(D) \neq Q_2(D)$, find a
  counterexample as a subinstance $D' \subseteq D$ such that
  $Q_1(D') \neq Q_2(D')$ and the size of $D'$ is minimized.} We
illustrate the setting with an example.


\begin{Example}
\label{example:intro1}
Consider the following two relation schema storing information about students and course registrations in a university: 
${\tt Registration(name, course, dept, grade)}$ and ${\tt Student(name, major)}$.
In a database course, suppose the instructor asked the students to write a SQL query to find students who registered for exactly one Computer Science (CS) course. The test instances $S, R$ of these two tables are given in Figure~\ref{fig:running-instance}. The following query $Q_1$ solves this problem correctly:
\\
\begin{figure}[t]
\subfigure[Table {\tt Student $S$}\label{fig:tab-student}]{
{\scriptsize
	\begin{tabular}{|c|c||c|}
	\hline
	{\tt name} & {\tt major} &  \\ \hline
	Mary & CS & $t_1$ \\
	John & ECON & $t_2$ \\ 
	Jesse & CS & $t_3$ \\
	\hline
    \end{tabular}
    }
}
\subfigure[Table {\tt Registration $R$}\label{fig:registration}]{
{\scriptsize
	\begin{tabular}{|c|c|c|c||c|}
	\hline
{\tt name} & {\tt course} & {\tt dept} & {\tt grade} &  \\ \hline
	Mary & 216 & CS & 100 & $t_4$  \\
	Mary & 230 & CS & 75 & $t_5$ \\
	Mary & 208D & ECON & 95 & $t_6$ \\
	John & 316 & CS & 90 & $t_7$ \\
	John & 208D & ECON & 88 & $t_8$ \\
	Jesse & 216 & CS & 95 & $t_9$ \\
	Jesse & 316 & CS & 90 & $t_{10}$ \\
	Jesse & 330 & CS & 85 & $t_{11}$ \\
	\hline
    \end{tabular}
    }
}
\caption{\label{fig:running-instance} Toy instances of tables in Example~\ref{example:intro1}. Identifiers are shown for all tuples.}
\vspace{-5mm}
\end{figure}
{\scriptsize
\begin{lstlisting}
$Q_1:$ SELECT s.name,s.major
FROM Student s, Registration r
WHERE s.name = r.name AND r.dept = 'CS'
EXCEPT
SELECT s.name,s.major
FROM Student s, Registration r1, 
	Registration r2
WHERE s.name = r1.name AND s.name = 
	r2.name AND r1.course <> r2.course AND
	r1.dept = 'CS' AND r2.dept = 'CS'
\end{lstlisting}
}
However, one student wrote $Q_2$, which actually finds students who registered for one or more CS  courses.
{\scriptsize
\begin{lstlisting}
$Q_2:$ SELECT s.name,s.major
FROM Student s, Registration r
WHERE s.name = r.name AND r.dept = 'CS'
\end{lstlisting}
}
\vspace{-1mm}
\begin{figure}[t]
\subfigure[Result of $Q_1$\label{fig:ans-Q1}]{
{\scriptsize
		\begin{tabular}{|c|c||c|}
		\hline
		{\tt name} & {\tt major} & \\ \hline
		John & ECON & $r_1$\\
		\hline
		\end{tabular}	
}
}
\subfigure[Result of $Q_2$\label{fig:ans-Q2}]{
{\scriptsize
		\begin{tabular}{|c|c||c|}
		\hline
		{\tt name} & {\tt major} & \\ \hline
		Mary & CS & $r_2$\\
		John & ECON & $r_3$\\
		Jesse & CS & $r_4$\\
		\hline
		\end{tabular}
		}
		}
		\caption{\label{fig:running-answers}Results of $Q_1, Q_2$ in Example~\ref{example:intro1}}
\vspace{-2mm}
\end{figure}
The results of queries $Q_1$ and $Q_2$ are given in Figure~\ref{fig:running-answers}.
The tuples $r_2 $ = {\tt(Mary, CS)} and $r_3$ = {\tt (Jesse, CS)} are in the output of $Q_2$ but not in the output of $Q_1$. 
To convince the student that his query is wrong, the instructor can provide the instances $S, R$ as a counter example comprising 11 tuples. 
However, a smaller and better counterexample can simply contain three tuples (\eg, $t_1, t_4, t_5$) to illustrate the inequivalence of $Q_1, Q_2$. The benefit will be much larger if we consider a real enrollment database from a university, whereas the size of the counterexample would remain the same.
\end{Example}


\par
Prior work in the database community mainly focused on the theoretical study of decidability \cite{nutt1998deciding,cohen2005equivalences} or generating a comprehensive set of test databases to ``kill'' as many erroneous queries as possible \cite{chandra2015data}, 
but does not pay much attention to explaining why two queries are inequivalent. 
There are recent systems that aim to generate counterexamples for SQL
queries. Cosette developed by Chu et al.\cite{chu2017cosette} used
formal methods that encodes SQL queries into logic formulas to
generate a counterexample that proves two SQL queries are
inequivalent. It generates counterexamples iteratively, so it must
return the smallest one. XData by Chandra et al.\cite{chandra2015data}
generates test data using mutation
techniques. 
However, counterexamples generated by such systems can lead to
arbitrary values, which may not be meaningful to the user.  Our
approach instead ensures that the user sees familiar values and
relationships already present in the test database instances.

\textbf{Our contributions.~} We make the following contributions in this paper. 
\begin{itemize}
\vspace{-1mm}
\item We formally define the smallest counterexample problem, and connect it to \emph{data provenance} with the definition of the smallest witness problem (Section~\ref{sec:prelim}).
\item We give complexity results (NP-hardness proofs and poly-time algorithms) in terms of both data and combined complexity for different subclasses of SPJUDA queries
(Section~\ref{sec:spjud}).
\item We give practical algorithms for SPJUD queries using SAT and SMT solvers, and discuss a suite of optimizations to improve the efficiency (Section~\ref{sec:general-algo}).
\item For aggregate queries, we illustrate the new challenges, and propose new approaches to address these challenges, which includes applying \emph{provenance for aggregate queries} \cite{amsterdamer2011provenance}, adapting the problem definition by parameterizing the queries, and rewriting the aggregate queries to reduce the number of tuples involved in the constraints to the SMT solver (Section~\ref{sec:algo_agg}).
\item We describe our implementation of the end-to-end
  \OurSys{}~system, which has been deployed in an undergraduate course
  (Section~\ref{sec:impl}).
\item We give extensive experimental results in Section~\ref{sec:experiments} to show how our approach can scale to large datasets (100K tuples for queries from the course and scale-1 for TPC-H queries). Also, we demonstrate that our optimizations reduce the size of the counterexample.
\item We provide a large, thorough user study from the undergraduate
  database course, where we let students use \OurSys{} to debug their
  RA queries in a homework. Quantitative analysis of usage statistics
  and homework scores shows that use of \OurSys{} improved student
  performance; anonymous survey of the students also indicates that
  they found \OurSys{} helpful to their learning
  (Section~\ref{sec:user-study}).
\end{itemize}

\cut{
Moreover, in practice, developers usually test their SQL queries by executing on large real-world data or large randomized generated data --- in this case, generating a counterexample from scratch may not be the best way for helping people understand the query inequivalence since it does not utilize the information from test database instances. Therefore, if people already know that two queries are inequivalent by testing against an database instance, the counterexample-generating methods have two major drawbacks in providing explanations. On one hand, knowing that the two queries are inequivalent with a test database instance ought to make it possible to find the smallest counterexamples faster: Cosette takes seconds on average to find counterexamples for inequivalent queries. On the other hand, the generated database instances may not be consistent with the original ones and thus may be the easiest way for people to understand: users are more familiar with their test database instances, also the values in the test instances may be meaningful but the generated counterexample is likely to be meaningless.

To address these challenges, first we formalized the smallest counterexample problem, discussed the hardness of the problem and proposed solutions for different query classes, and present a constraint solver based method. Based on data provenance, we can obtain logical formulas indicating whether a tuple is in the query result. For instance, given the logical formula $t_1 \ || \ (t_2\  \&\& \ t_3)$, a constraint solver\cite{de2008z3} can return a model $\{t_1: True, \ t_2: True, \ t_3: False\}$. And then we can build a database instance according to the model. 

One application scenario is to help people learn SQL queries: to grade SQL assignments in a database course, instructors usually run students' queries on a test input, and then compare their output with the output by a standard solution. However, students can only know ``correct'' or ``incorrect'' unless further comments are given. Returning counter examples will enable them to learn to write SQL interactively.
}

\vspace{-2mm}
\section{Preliminaries}
\label{sec:prelim}
We consider the class of Select (S)-Project (P)-Join (J)-Union(U)-Difference(D)-Aggregate(A) queries expressed as relational algebra (RA) expressions extended with aggregates. However, we will use RA form and SQL form of queries interchangeably. A subset of these operators using abbreviations will denote the corresponding subclass of such queries; \eg, PJ queries will denote queries involving only projection and join operations. 
\par
For a database instance $D$ (involving one or more relational tables) and a query $Q$, $Q(D)$ will denote the output of $Q$ on $D$. Let $\constraints$ denote a set of integrity constraints on the schema of the database instance $D$. We consider the following standard integrity constraints: keys, foreign keys, not null, and functional dependencies. If $D$ satisfies $\constraints$, we write $D \models \constraints$. We use $|D|$ to denote the total number of tuples in $D$.
\par
We will use unique identifiers to refer to the tuples in the database and query answers. In our example tables, they are written in the right-most column (see Figures~\ref{fig:running-instance} and \ref{fig:running-answers}), \eg, in Figure~\ref{fig:running-instance}, $t_1$ refers to the tuple ${\tt Student}(Mary, CS)$.

\subsection{Smallest Counterexample Problem}\label{sec:counterexample}



Consider two queries $Q_1$ and $Q_2$ such that $Q_1(D) \neq Q_2(D)$ on a database instance $D$ such that $D \models \constraints$ for a given set of integrity constraints $\constraints$. 
In other words, $D$ explains why $Q_1$ and $Q_2$ are inequivalent.
Based on $D$, we want to find a small counterexample $D' \subseteq D$ that also explains the inequivalence of $Q_1$ and $Q_2$.

\begin{Definition}[Counterexample and The Smallest Counterexample Problem]\label{def:SCP}
Given a database instance $D$, a set of integrity constraints $\constraints$ s.t. $D \models \constraints$, two queries $Q_1$, $Q_2$ where $Q_1(D) \neq Q_2(D)$, a counterexample is a subinstance $D' \subseteq D$ s.t. $D' \models \constraints$ and $Q_1(D') \neq Q_2(D')$. In particular, $D$ is a trivial counterexample.
\par 
The goal of the \emph{smallest counterexample problem  \\($\SCP(D, Q_1, Q_2)$)} is to find a counterexample $D' \subseteq D$ such that 
the total number of tuples in $D'$ is minimized  (\ie, for all counterexamples $D'' \subseteq D$, $|D''| \geq |D'|$).
\end{Definition}
In the above definition, we assume that the results of the two queries are \emph{union-compatible} (\ie, $Q_1(D), Q_2(D)$ have the same schema), which is easy to check syntactically (otherwise the difference in their schema serves as the reason of their inequivalence). 

Note that 
keys, functional dependencies, and \emph{not null} constraints  are closed under subinstances, \ie, for such constraints $\constraints$, if $D \models \constraints$, then $\forall D' \subseteq D$, $D' \models \constraints$.
Therefore, for such constraints, no additional consideration is needed. This is not true for referential constraints or foreign keys, which we explicitly consider in our algorithms. From now on, where it is clear from the context, we will implicitly assume that the $D' \subseteq D$ discussed as counterexamples satisfy the given constraints $\constraints$.

\begin{Example} 
\label{example:smallest-1}
In Example~\ref{example:intro1} and Figure~\ref{fig:running-instance}, the given test instances $S$ and $R$ of input relations {\tt Student} and {\tt Registration} already form a counterexample for $Q_1$ and $Q_2$. However, some subinstances of $S, R$ are also counterexamples. Among these subinstances, $S'$ $=\{t_1\}$, $R'$ $= \{t_4,t_5\}$; or $S''$ $=\{t_3\}$, $R''$ $= \{t_9,t_{10}\}$ are two smallest counterexamples (there are two other smallest counter examples varying the two courses of Jesse), i.e. there are no counterexamples with less than 3 tuples.
\end{Example}


Our goal is to explain the query inequivalence to users by showing the smallest counterexample over which the two queries return different results. Even in our running example with a toy database instance, this reduced the number of tuples from 11 to only 3, whereas the benefit is likely to be much more for test database instances in practice as observed in our experiments. The brute-force method to find the smallest counterexample is to enumerate all subinstances of $D$, and search for the smallest subinstance $D'$ where $Q_1(D')$ and $Q_2(D')$ are different. However, enumerating all possible subinstances is inefficient and it does not utilize the information that $D$ is already a counterexample. Therefore, to solve this problem more efficiently, we relate this problem to the concepts of \emph{witnesses} and  \emph{data provenance} as discussed in the next two subsections.

\subsection{Smallest Witness Problem}\label{sec:witness}


Buneman et al. \cite{buneman2001and} proposed the concept of \emph{witnesses} to capture \emph{why-provenance} of a query answer. Intuitively, a witness is a collection of input tuples that provides a \emph{proof} for a given output tuple. Formally, given a database instance $D$, a query $Q$, and a tuple $t$ $\in Q(D)$, a witness for $t$ w.r.t. $Q$ and $D$ is a subinstance $D' \subset D$ where $t \in Q(D')$. For instance, in Example~\ref{example:intro1}, $\{t_1, t_4\}$, $\{t_1, t_5\}$, and $\{t_1, t_4, t_5\}$ are three witnesses of the output tuple $r_2$ w.r.t. $Q_2$ and $D$. We use $\witset(Q, D, t)$ to denote the set of all witnesses for $t \in Q(D)$ w.r.t. $Q$ and $D$.

In the smallest counterexample problem $\SCP(D, Q_1, Q_2)$, since $Q_1(D) \neq Q_2(D)$, there must exist a tuple $t$ such that $t \in Q_1(D) \setminus Q_2(D)$, or, $t \in Q_2(D) \setminus Q_1(D)$. Since $Q_1, Q_2$ are assumed to be union-compatible,  we can construct two queries $Q^{'}_1 = Q_1 - Q_2$ and $Q^{'}_2 = Q_2 - Q_1$. Therefore, for any counterexample $D' \subseteq D$ for $Q_1$ and $Q_2$, $\exists t$ such that $t \in Q^{'}_1(D')$, or, $t \in Q^{'}_2(D')$. Given such an answer tuple $t$ differentiating $Q_1, Q_2$, we say that $D'$ \emph{witnesses} the tuple $t$ in the result of $Q^{'}_1$ or $Q^{'}_2$. 

A witness may contain many tuples and is sensitive to the query structure. Buneman et al. \cite{buneman2001and} defined minimal witness as a minimal element of $\witset(Q, D ,t)$, \ie, for a minimal witness $w \in \witset(Q, D, t)$, there exist no other witnesses $w' \in \witset(Q, D, t)$ such that $w' \subset w$. In Example~\ref{example:intro1}, $\{t_1, t_4\}$ and $\{t_1, t_5\}$ are minimal witnesses of the output tuple $r_2$ w.r.t $Q_2$ and $D$, but $\{t_1, t_4, t_5\}$ is not. In particular, a witness with the smallest cardinality must be a minimal witness.






\begin{Definition}[Smallest Witness Problem]\label{def:SWP}
Given a database instance D, two union-compatible queries $Q_1$ and $Q_2$ s.t. $Q_1(D) \neq Q_2(D)$, and a tuple $t$ s.t. $t \in Q_1(D) \setminus Q_2(D)$ or $t \in Q_2(D) \setminus Q_1(D)$, the goal of the smallest witness problem ($\SWP(D, Q_1, Q_2, t)$) is to find a witness $w \in \witset(Q_1 - Q_2, D, t) \cup \witset(Q_2 - Q_1, D, t)$ such that the total number of tuples in $w$ is minimized.
\end{Definition}

We can reduce the smallest counterexample problem \\$\SCP(D, Q_1, Q_2)$ into the smallest witness problem $\SWP(D, Q_1,$ \\$Q_2, t)$ by enumerating all possible output tuples in the difference of $Q_1(D)$ and $Q_2(D)$,  solving $\SWP(D, Q_1, Q_2, t)$, and finding the globally minimum witness across all such $t$-s.  
$$
\SCP(D, Q_1, Q_2) = \min_{t \in (Q_1(D) \setminus Q_2(D)) \cup (Q_2(D) \setminus Q_1(D))} \SWP(D, Q_1, Q_2, t)\\
$$

From now on, without loss of generality, we will assume that in the smallest counterexample problem $\SCP(D, Q_1, Q_2)$, there exists a tuple $t \in Q_1(D)$ but $t \notin Q_2(D)$. In the rest of the paper, we will mainly focus on the smallest witness problem $\SWP(D, Q_1, Q_2, t)$ for such a tuple, primarily due to the fact that it provides more efficient solutions and allows optimizations compared to \SCP.  We further discuss the connection between \SCP\ and \SWP\ in Section~\ref{sec:general-algo} and in \cut{Appendix~\ref{sec:scp-vs-swp}.}Section ~\ref{sec:experiments}.




\subsection{Boolean How-Provenance}\label{sec:how-prov}


Buneman et al.\cite{buneman2001and} formally introduced the why-provenance model that captures witnesses for a tuple $t$ in the result of a query $Q$ on a database instance $D$. However, 
it lacks an efficient method to compute the smallest witness from why-provenance. In order to compute the smallest witness efficiently for general SPJUD queries, we use the concept of \emph{how-provenance} or \emph{lineage} \cite{green2007provenance,amsterdamer2011limitations}. How-provenance encodes how a given output tuple is derived from the given input tuples using a Boolean expression, and its first use can be traced back to 
Imilienski and Lipski \cite{Imielinski:1984:IIR:1634.1886} who used it to describe incomplete databases or c-tables.
The computation of how-provenance of an output tuple $t \in Q(D)$, denoted by $\Prv_{Q(D)}(t)$ or $\Prv(t)$ when clear from the context, is well known and intuitive: tuples in the given input relations are annotated with unique identifiers (as shown in the right-most columns in Figure~\ref{fig:running-instance}). As the query $Q$ executes, for joint usages of sub-expressions (joins), their annotations are combined with conjunction ($\wedge$ or $\cdot$), and for alternative usages of sub-expressions (projections or unions), the annotations are combined with disjunction ($\vee$ or $+$). For simplicity, we use $+$ for disjunction, and omit symbols for conjunction. For instance, in Example~\ref{example:intro1}, in $Q_2(D)$,
\begin{equation}
\Prv_{Q_2(D)}(r_2) = t_1t_4 + t_1t_5 = t_1(t_4 + t_5) = \phi_1  {\rm (say)} \label{equn:prov-r2}
\end{equation}
For set difference operation, consider $R = R_1 - R_2$, where all tuples in $R_1, R_2$ are annotated with  how-provenance. If a tuple $t$ appears in $R$, it must appear in $R_1$. Suppose $\Prv_{R_1}(t) = \phi$. If $t$ does not appear in $R_2$, $\Prv_R(t) = \phi$. If $t$ does appear in $R_2$ with $\Prv_{R_2}(t) = \psi$, then $\Prv_{R}(t) = \phi \cdot \overline{\psi}$, where $\overline{\psi} = \neg \psi$ denotes the negation of the Boolean expression $\psi$. This implies $t$ appears in the final results of $R$ if $t$ appears in $R_1$ but not in $R_2$. 

\begin{example}\label{eg:queryplan}
In Example~\ref{example:intro1}, consider the following RA expressions for $Q_2$ and $Q_1$, using abbreviations ${\tt S}$ and ${\tt R}$ for ${\tt Students}$ and ${\tt Registration}$, where $\Join$ denotes natural join (abusing the form of RA for simplicity).
\begin{equation}
Q_2  = {\tt \pi_{name, major} \sigma_{dept = 'CS'} (S \Join R)}\label{equn:Q2}
\end{equation}
Suppose $Q_3 = {\tt \pi_{name, major} \sigma_{\eta} (S \Join R ~r1 \Join R ~r2)}$, where $\eta$ denotes the selection condition:  $r1.dept = 'CS' \wedge r2.dept = 'CS' \wedge r1.course != r2.course$. Then 
$Q_1  = Q_2 - Q_3$. 
Consider the result tuple $r_2 = (Mary, CS)$, which is in $(Q_2 - Q_1)(D)$ (Figure~\ref{fig:running-answers}). The provenance of $r_2 = (Mary, CS)$ in $Q_2(D)$ is given in Equation~(\ref{equn:prov-r2}). It does not appear in $Q_1(D)$ since it appears in both $Q_2, Q_3$ in (\ref{equn:Q2}). For $Q_3$,
$\Prv_{Q_3(D)}(r_2) = t_1t_4t_5 = \phi_2 {\rm (say)}$.
Hence, $\Prv_{Q_1(D)}(r_2)$ $=\phi_1 \cdot \overline{\phi_2}$, and 
$\Prv_{(Q_2-Q_1)(D)}(r_2)$
$ = \phi_1 \cdot \overline{[\phi_1 \cdot \overline{\phi_2}]}$
$=  \phi_1 \cdot [\overline{\phi_1} + \phi_2]$
 $ =  \phi_1 \cdot \phi_2$
 $ =  (t_1(t_4 + t_5)) \cdot (t_1t_4t_5)$
$ =  t_1t_4t_5$.
In other words, the tuple $(Mary, CS)$ can distinguish the queries $Q_1, Q_2$ in a small witness $S' = \{t_1\}, R' = \{t_4, t_5\}$, which solves both $\SWP$ and $\SCP$ problems. 
\end{example}
For the above example, the smallest witness or the smallest counterexample could be found by inspection, since $Q_1, Q_2$ are similar. For arbitrary and more complex queries, how-provenance gives a systematic approach to find a small witness as we will discuss in the following two sections.

\cut{
\par
\textbf{Provenance to compute smallest witness.~} \red{may be merged with section 4..} Instead of $(Mary, CS)$, if we considered $r_4 = (Jesse, CS) \in (Q_2 - Q_1)(D)$, then it can be verified that $\phi_1' = t_3(t_9 + t_{10} + t_{11})$, $\phi_2' = t_3(t_9t_{10} + t_{9}t_{11} + t_{10}t_{11})$, and $\Prv_{(Q_2-Q_1)(D)}((Jesse, CS))$ $ = \phi_1' \cdot \phi_2' = \phi_2'$ $= t_3t_9t_{10} + t_3 t_{9}t_{11} + t_3t_{10}t_{11}$. Looking at the \emph{minterms} of this DNF expression, we get three choices of smallest witnesses of $(Jesse, CS)$, where ${\tt Student}$ contains $t_3$, and ${\tt Registration}$ contains two tuples out of $t_9, t_{10}, t_{11}$. This gives an algorithm, albeit not necessarily efficient, to solve the $\SWP(D, Q_1, Q_2, t)$ problem: (i) Compute the how-provenance of $t$, $\phi_t = \Prv_{(Q_1-Q_2)(D)}(t)$; (ii) Expand $\phi_t$ into the equivalent DNF form, (iii) Apply absorption ($a + ab = a$), commutativity ($ab = ba$), and idempotence ($a+a = a$), (iv) Return any minterm of the resulting solution as the smallest witness. While the provenance as an arbitrary Boolean  formula can be computed in polynomial time (and is of poly-size), expansion into DNF may be  exponential  in the size of the database $|D|$ if the queries are not monotone. 
}


\textbf{Aggregates.~} In the next two sections, we discuss algorithms and complexity results for SPJUD queries. As we discuss in Section~\ref{sec:algo_agg}, aggregate queries entail new challenges, where we adapt the definitions of optimization problems accordingly and discuss solutions.

\cut{
for $Q_1$,  even $(Mary, CS)$ is not in the query result, $\Prv((Mary, CS)) = ((t_1 + t_4) (t_1  t_5)) \ \overline{t_1  t_4  t_5}$. We call $Prv(t)$ the provenance expression of t. 
A tuple $t$ is ``witnessed'' in the result of $Q_1-Q_2$ on $D$ if and only if $\Prv_{(Q_1 - Q_2)(D)}(t)$ is true.
}

\cut{
 proposed an approach using the how-provenance model to capture how a result tuple was derived from the input tuples. Using the semiring framework we can encode the how-provenance that describes how an output tuple $t$ has been derived from source or intermediate tuples, denoted as $Prv(t)$, using joins (logical and $\land$), union (logical or $\lor$), set-difference(logical and with negation: $a - b$ $=a \land \neg b$) operations. And to reduce the size of expressions we used $+$ for $\lor$, $\cdot$ for $\land$ and $\overline{b}$ for $\neg b$, and even omit $\cdot$ in most terms. 
 }

\cut{
copy-paste from Sudeepa's vldb'11 paper:

The annotation of tuples with boolean expressions goes
back to the c-tables of Imieli´nski and Lipski [23] who used
them to describe incomplete databases. It was then used
in [18, 44] and ever since for probabilistic databases. This is
a particular case of provenance annotation [21] which is why
in this paper we use the terminology “boolean provenance”.The algorithm [23] for computing boolean provenance is
quite well-known so we only illustrate it here on an example,
namely for the query q in Figure 1b and the instance
I in the same figure. We use the relational algebra expression
Πx(R ✶ ΠyS) which is equivalent to q and we show
the steps of the computation in Figure 2. First, to compute
ΠyS, the boolean provenances of tuples in S with the
same value of y are combined using disjunction (+) (see Figure
2a); disjunction is also used with union. For join, the
boolean provenance of two tuples are combined using conjunction
(see Figure 2b). Thus, the boolean provenance of
the answers of positive queries is a monotone boolean expression.
Now consider R = R1 − R2 where the tuples in R1, R2
are annotated with boolean provenance. For a tuple t to
appear in R, t must appear in R1, let’s say with boolean
provenance φ. If t does not appear in R2, then the boolean
provenance of t in R will be also φ; on the other hand if
t appears in R2, let’s say with boolean provenance ψ, then
the boolean provenance of t in R will be φ ∧ ψ. Figure 2d
shows the boolean provenance for the difference query q
0 −q
where q, q0
are given in Figure 1
}

\section{Complexity for SPJUD Queries}
\label{sec:spjud}

\cut{In this section, we consider the problem of finding the smallest counterexample for two SPJUD queries $Q_1$ and $Q_2$ with a test database instance $D$. We assume that $Q_1$ and $Q_2$ are union-compatible and of the same query class, and there exists a tuple $t \in Q_1(D) - Q_2(D)$.
We solve this problem by finding the smallest witness for a output tuple $t$ in the query result of $Q_1(D)-Q_2(D)$}

Table~\ref{table:spjud-complexity} summarizes the complexity of the smallest witness problem (\SWP) for any subclass of SPJUD queries. In terms of complexity, we consider \emph{data complexity} (fixed query size), \emph{query complexity} (fixed data size), and \emph{combined complexity} (in terms of both data and query size) \cite{vardi1982complexity}. 
Thus polynomial combined complexity indicates polynomial data complexity.  
\cut{We use $n = |D|$ to denote the total number of tuples (the parameter for data complexity), $k$ to denote the number of relations involved in the query (the parameter for query complexity), and assume each relation has constant number of attributes. }

\begin{table}[htb]
{\small
\begin{tabular}{|p{5em}|p{8em}|p{8em}|} \hline
\rowcolor{grey} \shortstack[l]{Query Class \\ of $Q_1, Q_2$} & \shortstack[l]{Data \\Complexity} & \shortstack[l]{Combined\\ Complexity} \\ \hline
SJ & {\bf P} (Thm.~\ref{thm:sj-poly})& {\bf P} (Thm.~\ref{thm:sj-poly}) \\ \hline
SPU & {\bf P} (Thm.~\ref{thm:spu-poly}) & {\bf P} (Thm.~\ref{thm:spu-poly}) \\ \hline
PJ & {\bf P} (Thm.~\ref{thm:SPJU-poly}) & {\bf \NPhard{}} (Thm.~\ref{thm:PJ-hard}) \\ \hline
JU & {\bf P} (Thm.~\ref{thm:SPJU-poly})& {\bf \NPhard{}} (Thm.~\ref{thm:JU-hard})\\ \hline
JU$^*$ & {\bf P}  (Thm.~\ref{thm:JU-restricted-poly}) & {\bf P} (Thm.~\ref{thm:JU-restricted-poly}) \\ \hline
SPJUD$^*$ & {\bf P}  (Thm.~\ref{thm:SPJUD-restricted-poly}) & {\bf \NPhard{}~}if falls into class PJ or JU \\ \hline
PJD & {\bf \NPhard{}} (Thm.~\ref{thm:SPJUD-hard}) & {\bf \NPhard{}} (Thm.~\ref{thm:SPJUD-hard}) \\
\hline
\end{tabular}
}
\caption{Complexity dichotomy of finding smallest witness for a result tuple w.r.t. the difference of two queries $Q_1 - Q_2$. The class JU$^*$ has the restriction that all unions appear after all joins. The class SPJUD$^*$ is defined as: $Q \rightarrow q^+ | Q - Q$, where $q^+$ is a terminal that represents SPJU queries. Proofs are given in Appendix~\ref{sec:spjud-proof}.}
\label{table:spjud-complexity}
\end{table}

For queries involving PJ, in general even the query evaluation problem is NP-hard in query complexity. However, we construct acyclic queries that can be evaluated in poly-time in combined complexity. It's the same for queries involving JU, however, the problem is in poly-time for the subclass JU$^*$, because we can directly look into the join-only parts of a JU$^*$ query. For general SPJU queries, the problem has poly-time data complexity, and thus we can provide a poly-time algorithm for SPJUD$^*$ queries in data complexity.

What is noteworthy is that for the class of queries involving projection, join, and difference, it is already \NPhard{} in data complexity to find the smallest witness for a result tuple; and the result holds even when the queries are of bounded sizes and the database instance only contains two relations. 

While in the complexity results, we assume both $Q_1, Q_2$ belong to the same query class, if $t \in Q_1(D)\setminus Q_2(D)$, for all monotone cases the exact class of $Q_2$ does not matter as long as it is monotone. 

\section{A Constraint-based General Solution for SPJUD Queries}
\label{sec:general-algo}
In the previous section, we showed that for a number of query classes, the smallest witness problem is poly-time solvable in data complexity. However,  the problem is still \NPhard{} in general, even when the queries are of bounded size; further, the poly-time algorithms we discussed are not efficient for practical purposes.
\cut{--- not only traversal in the query trees is required, but we have to test all subinstances during the enumeration process.}
To address these challenges, we introduce a constraint-based approach to the smallest witness problem. We map the problem into the \emph{min-ones satisfiability problem} \cite{kratsch2010parameterized} by tracking the Boolean provenance of output tuples. The min-ones satisfiability problem is an extension of  the classic \emph{Satisfiability (SAT)} problem: \emph{given a Boolean formula $\phi$, it checks whether $\phi$ is satisfiable with at most $k$ variables set to true.}
This problem can be solved by either using a \emph{SAT solver} (\eg, MiniSAT\cite{sorensson2009minisat}, \cut{Lingeling\cite{biere2013lingeling},} and CaDiCaL\cite{CaDiCaL2018}), or an \emph{SMT Solver} (\eg, CVC4\cite{barrett2011cvc4}and Z3\cite{de2008z3}\cut{, Yices\cite{dutertre2014yices}, and SMTInterpol\cite{christ2012smtinterpol}}). \emph{Satisfiability Modulo Theories (SMT)} is a form of \emph{constraint satisfaction problem}. It refers to the problem of determining whether a first-order formula is satisfiable w.r.t. other background first-order formulas, and is a generalization of the SAT problem\cite{barrett2018satisfiability}). SAT and SMT problems are known to be \NPhard{} with respect to the number of clauses, constraints, and undetermined variables. However, there is  a variety of solvers that work very well in practice for different real world applications, and with these solvers we can find a small solution to a \SWP \ instance.
The rest of this section will describe how to encode the how-provenance of an output tuple, and then use a state-of-the-art solver to find the smallest witness for the output tuple. The implementation details will be discussed in Section~\ref{sec:impl}.


\vspace{-2mm}
\subsection{Passing How-Provenance to a Solver}\label{sec:prov-solver}


As discussed in Section~\ref{sec:how-prov}, the how-provenance $Prv(t)$ is true $\Leftrightarrow$ tuple $t$ is in the query result. Since $Prv(t)$ is composed of a combination of Boolean variables annotating tuples in the input relations, a Boolean variable is true $\Leftrightarrow$ the corresponding tuple is present  in the input relation in the witness. Then an instance of the smallest witness problem is mapped to an instance of the min-ones satisfiability problem: find a satisfying model to $Prv(t)$ with least number of variables set to true, and the variables set to true in the satisfying model indicate tuples in the smallest witness. The pseudocode of the algorithm to solve \SCP\ and \SWP\ is given in Algorithm~\ref{alg:naive-sat}.

Example~\ref{example:how-provenance-sat-solver} illustrates how we can get the smallest witness using a SAT solver. Since the solver will return an arbitrary satisfying model, to get the minimum model we need to ask the solver to return a different model every time we rerun it (line 6). We set a maximum number of runs to limit the running time, and the algorithm stops when there is no more satisfying models or it has reached the maximum number of runs. It may not find the minimum model when it stops, but 
it is likely to find one that is small if given enough time.

\begin{Example}[How-provenance and SAT Solver]
\label{example:how-provenance-sat-solver}
Consider Example \ref{example:intro1}.

\cut{
\begin{table}[!htb]
{\small\upshape
\caption*{How-provenance of $Q_1(D)$:}
\begin{center}
	\begin{tabular}{|p{3em}|p{3em}|p{15em}|}
	\hline
	name & major & prov \\ \hline
	Mary & CS & $ (t_1  t_4 + t_1 t_5) \overline{t_1 t_4 t_5} $\\
	John & ECON & $ t_2 t_7 $\\
	Jesse & CS & $(t_3 t_{9} + t_3t_{10} + t_3 t_{11}) \ \overline{t_3 t_9 t_{10} + t_3t_9 t_{11} +}$ $\overline{t_{3} t_{10} t_{11}} $\\
	\hline
	\end{tabular}
\end{center}
}
 \vspace{-3mm}
\end{table}
\begin{table}[!htb]
{\small\upshape
\caption*{How-provenance of $Q_2(D)$:}
\begin{center}
	\begin{tabular}{|p{3em}|p{3em}|p{15em}|}
	\hline
	name & major & prov \\ \hline
	Mary & CS & $ t_1 (t_4 + t_5) $\\
	John & ECON & $ t_2 t_7 $\\
	Jesse & CS & $t_3 (t_9 + t_{10} + t_{11})$\\
	\hline
	\end{tabular}
\end{center}
}
 \vspace{-2mm}
\end{table}}

(Mary, CS) and (Jesse, CS) is in the result of $Q_2$ but not in the result of $Q_1$, and the how-provenance for them w.r.t. $Q_2 - Q_1$ and $D$ can be computed based on these results. E.g., 
\cut{$\Prv_{(Q_2-Q_1)(D)}(Jesse, CS)$ = $\Prv_{Q_2(D)}(Jesse, CS) \land \neg$\\ $\Prv_{Q_1(D)}(Jesse, CS)$ = $(t_3 (t_9 + t_{10} + t_{11}))$ \\ $\overline{(t_3 t_{9} + t_3t_{10} + t_3 t_{11}) \\ \overline{t_3 t_9 t_{10} + t_3t_9 t_{11} + t_{3} t_{10} t_{11}}}$ $= t_3 t_9 t_{10} + t_3 t_9 t_{11}  + t_3 t_{10} t_{11} $. }
$\Prv_{(Q_2-Q_1)(D)}(Jesse, CS)$ = $\Prv_{Q_2(D)}(Jesse, CS) \land $\\ $\neg\Prv_{Q_1(D)}(Jesse, CS)$ = $(t_3 (t_9 + t_{10} + t_{11}))$ $\overline{(t_3 t_{9} + t_3t_{10} + t_3 t_{11})}$   \\$\overline{\overline{t_3 t_9 t_{10} + t_3t_9 t_{11} + t_{3} t_{10} t_{11}}}$ $= t_3 t_9 t_{10} + t_3 t_9 t_{11}  + t_3 t_{10} t_{11} $. 
Then we can get a model $\{t_3: True, t_9: True, t_{10}: True, t_{11}: True\}$ 
to $\Prv_{(Q_2-Q_1(D)}(Jesse, CS)$ by passing it to a SAT solver. We will get any of $\{t_3, t_9, t_{10}\}$, $\{t_3, t_9, t_{11}\}$, and $\{t_3, t_{10}, t_{11}\}$ as the smallest witness after running the solver for multiple times (in the first run, it may return a bigger solution like $\{t_3, t_9, t_{10}, t_{11}\}$).
\end{Example}

\begin{algorithm}[t]\caption{Basic: The SAT-solver-based approach}\label{alg:naive-sat}
{\small
\begin{codebox}
\Procname{$\proc{Smallest-Witness-Basic}(Prv(t), \Delta)$}
\li $\phi \gets Prv(t)$
\li $\eta^* \gets$ null
\li $\delta \gets 0$
\li \While {$\phi$ is satisfiable and $\delta < \Delta$}
\li \Do Use a SAT solver to find a model $\eta$ for $\phi$
\li 	$\phi \gets \phi \land \neg\eta$ 
\li 	\If $\eta^*$ is null or \# true variables in $\eta$ is less than $\eta^*$
		\Then
\li 		$\eta^* \gets \eta$
		\End
\li 	$\delta \gets \delta + 1$ 
	\End
\li Return a set of tuples $D_{\eta^*} \gets \{t' \mid \eta^*(t') \ is \ true \}$
\end{codebox}
\begin{codebox}
\Procname{$\proc{Smallest-Counterexample-Basic}(Q_1, Q_2, D, \Delta)$}
\li $D^* \gets D$
\li $\Delta$ is the maximum number of trials.
\li \For {$t \in (Q_1 - Q_2)(D)$ }
\li \Do	$Prv(t) \gets $ the how-provenance of tuple $t$ w.r.t. \\ $(Q_1 - Q_2)(D)$.
\li 	$D' \gets $ \proc{Smallest-Witness-Basic}$(Prv(t), \Delta)$
\li 	\If $|D'| < |D^*|$
		\Then
\li 		$D^* \gets D'$
		\End
	\End
\li Return $D^*$
\end{codebox}
}
\end{algorithm}

\vspace{-3mm}
\subsection{Optimizing the Basic Approach}\label{sec:opt-basic}
\label{subsec:spis}
The basic algorithm  given in Algorithm~\ref{alg:naive-sat} has two limitations: (a) it cannot find the smallest witness until it searches all possible models that satisfy $Prv(t)$; (b) In order to solve \SCP, it iterates over all tuples in $Q_1(D) \setminus Q_2(D)$ and calculates the provenance for each tuple, which leads to large overheads. Therefore, we propose two optimizations. The first one is to pick only one tuple $t$ from the query results of $Q_1(D) \setminus Q_2(D)$ (\ie, we only solve \SWP), and only compute the provenance of $t$ by adding an additional selection operator to select tuples equal to $t$ on top of the query tree of $Q_1 - Q_2$. 
The other optimization is to treat this problem as an optimization problem instead of finding different models with a SAT or SMT solver. However, integer linear programming solvers can not be applied because transforming how-provenance into linear constraints can be exponential. To solve this problem, we use \emph{optimizing SMT solvers} that are now available with recent advances in the programming languages and verification research community \cite{bjorner2015nuz,li2014symbolic}. Given a formula $\phi$ and an objective function $\mathcal F$, an optimizing SMT Solver finds a satisfying assignment of $\phi$ that maximizes or minimizes the value of $\mathcal F$. 

Algorithm~\ref{alg:spis} describes the solution with these two optimizations. A selection operator on the value of $t$ is added to $Q_1 - Q_2$ (line 2-3). Again, we add $Prv(t)$ as the constraint of the optimizing SMT solver, set the number of true variables as the objective function, and get the optimal model (line 4-6). Our SMT formulation includes only Boolean variables, so we encode the number of true variables by first converting the variables into 0 or 1 and then summing them up. 
\par
The SQL query optimizer is likely to push down the additional selection operator to accelerate the computation of how-provenance. Moreover, since a how-provenance may involve many tuples, solving it with an optimizer will reduce the solving time significantly, since the optimizer will return an answer as soon as it finds a solution, but the naive algorithm requires enumerating all possible models to obtain the model with least number of variables set to true.


\begin{algorithm}[t]\caption{$Opt_{\sigma}$: The optimized algorithm with selection pushdown}\label{alg:spis}
{\small
\begin{codebox}
\Procname{$\proc{Smallest-Counterexample-Optimized}(Q_1, Q_2, D)$}
\li Pick one tuple $t$ in the result of $Q_1(D) \setminus Q_2(D)$, 
\li $A_1 ... A_k \gets$ the attributes of $t$.
\li $Q' \gets \sigma_{A_1 = t.A_1, A_2 = t.A_2, ..., A_k = t.A_k} (Q_1-Q_2)$
\li $\phi \gets $ the how-provenance of tuple $t$ w.r.t. $Q'(D)$.
\li $obj \gets$ the number of {\tt true} values in $\phi$
\li $\eta \gets$ OptSMT\_Solver($\phi$, $obj$)
\li Return a set of tuples $D_{\eta} \gets \{t' \mid \eta(t') \ is \ true \}$
\end{codebox}
}
\end{algorithm}

\vspace{-1mm}
\begin{lstlisting}[language=smtlib2,style=smtlib2,caption={SMT-LIB Input for Example~\ref{example:how-provenance-sat-solver}},label={sample-smtlib2},captionpos=b]
(declare-const t1 Bool)
...
(declare-const t11 Bool)
(define-fun b2i ((x Bool)) Int (ite x 1 0))
(assert (and (or t4 t5) (not (and (or (and t1 t4) (and t1 t5)) (not (and t1 (and t4 t5)))))))
(minimize (+ (b2i t1) (b2i t2) ... (b2i t11)))
\end{lstlisting}
The above listing illustrates how we encode the provenance and constraints into the SMT-LIB standard format \cite{barrett2010smt} as the input to a SMT solver to find the satisfying model for Example~\ref{example:how-provenance-sat-solver}. In the sample SMT-LIB format input above, first we defined Boolean variables for each tuple from line 1 to line 4, then at line 4 we defined function b2i to convert Boolean variables for each tuple into 0 and 1. At line 5 we added the how-provenance as a constraint. Then with function b2i we take the sum of 0-1 variables to get the number of true variables in the model, and set the sum as the objective function (line 6).

\subsection{Handling Database Constraints}
\label{subsec:constraints}

Since we output a subinstance of the input database instance as the witness, database constraints like keys, not null, and functional dependencies are trivially satisfied if the input instance is valid. On the other hand, foreign key constraints can be naturally represented as Boolean formulas like provenance expressions. For instance, in our running example in Figure~\ref{fig:running-instance}, the {\tt name} column in the {\tt Registration} table may refer to the {\tt name} column in the {\tt Student} table. So, if we want to keep any tuple in the {\tt Registration} table, we must also keep the tuple with the same {\tt name} value in the {\tt Student} table. This constraint can be expressed in the $a \Rightarrow b$ form, \eg, $t_1 + \overline{t_4}$, $t_2 + \overline{t_7}$, .., etc., corresponding to the constraint that the tuples in the {\tt Registration} table cannot exist unless the tuple it refers to exists in the {\tt Student} table). Then, for each tuple that appears in the provenance expression added to the SAT or SMT solver, we add its foreign key constraint expression to the solver as a constraint. 

\section{Aggregate Queries}

\label{sec:algo_agg}

So far, we have focused on SPJUD queries. In this section we extend our discussion to aggregate queries. First we will demonstrate the challenges that arise for aggregate queries, and then propose our solutions to overcome them. We make some assumptions on the form of aggregate queries: (1) no aggregate values or NULL values are allowed in the group by attributes; (2) selection predicates involving aggregate values (HAVING) are in the simple form $expr \ relop \ expr$; (3) there is no difference operation above an aggregate operator.

\subsection{Challenges with Aggregate Queries}\label{sec:challenges-agg}


\mypar{Witness is too strict} Remember that for SPJUD queries, we find the smallest counterexample by first picking an output tuple in $(Q_1-Q_2)(D)$ and then finding the smallest witness for this tuple w.r.t. $Q_1-Q_2$ and $D$. However, for aggregate queries, if we still look for witnesses for output tuples, it is likely that we are unable to find any witnesses smaller than the input database instance --- the aggregate value may change if any tuple is removed from the input. 
Following Example~\ref{example:agg1} illustrates this issue.

\begin{Example}[Challenge with Witness for Aggregate Values]
\label{example:agg1}
Suppose we have two aggregate queries $Q_1$ and $Q_2$ aimed at computing the average grade of students in CS courses, using the two tables in Figure~\ref{fig:running-instance}.

\begin{lstlisting}[basicstyle=\scriptsize\upshape\ttfamily]
$Q_1:$SELECT s.name, avg(r.grade) as avg_grade
FROM Student s, Registration r 
WHERE s.name=r.name AND r.dept='CS'
GROUP BY s.name
\end{lstlisting}
\vspace{-3mm}
\begin{lstlisting}[basicstyle=\scriptsize\upshape\ttfamily]
$Q_2:$SELECT s.name, avg(r.grade) as avg_grade
FROM Student s, Registration r 
WHERE s.name=r.name
GROUP BY s.name
\end{lstlisting}
\vspace{-5mm}
\begin{table}[!htb]
\scriptsize{
    \begin{minipage}{.4\linewidth}
	    \centering
	   	\caption*{Result of $Q_1(D)$:}
	    \begin{tabular}{| l | l |}
		\hline
		name & avg\_grade \\ \hline
		Mary & 87.5 \\
		John & 90 \\
		Jesse & 92.5 \\
		\hline
		\end{tabular}
	\end{minipage}
	\begin{minipage}{.4\linewidth}
		\centering
		\caption*{Result of $Q_2(D)$:}
		\begin{tabular}{| l | l |}
		\hline
		name & avg\_grade \\ \hline
		Mary & 90 \\
		John & 89 \\
		Jesse & 92.5 \\
		\hline
		\end{tabular}
	\end{minipage}
}
 \vspace{-2mm}
\end{table}
In this example, $Q_2$ forgets to select on departments. To find a counterexample through finding a witness, we have to keep all records of the student as the witness. E.g., we can only return all Mary's registration records as the witness $W$ to keep (Mary, 90) in $Q_2(W)$ but not in $Q_1(W)$. However, to show that $Q_1$ will return a different result from $Q_2$ over \emph{some} counterexample $C$, $C$ can contain only one tuple $(Mary, 208D, ECON, 88)$, and $Q_1(C)$ is empty while $Q_2(C)$ returns $(Mary, 88)$.
\end{Example}
\mypar{Computation overhead by how-provenance} We cannot directly apply \spjudBasic{} or \spjudOpt{} (Section~\ref{sec:general-algo}) for aggregate queries because: (i) the why-provenance model does not consider aggregate queries; (ii) while how-provenance can be extended to support aggregate queries by storing all possible combinations of grouping tuples \cite{sarma2008exploiting}, it leads to exponential overhead and thus is impractical if there exist large groups. 

\mypar{Selection predicates with aggregate values} When the queries contain selection predicates with aggregate functions COUNT or SUM, it is possible that we have to keep all tuples in one group to make the result tuple satisfy the selection predicates. See Example~\ref{example:selection-agg-2}.

\begin{Example}[Challenge with Selection on Aggregate Values]

\label{example:selection-agg-2}

Continued with Example~\ref{example:agg1}, but both queries are extended to find the average grade of CS courses of students who registered at least 3 CS courses.

\begin{lstlisting}[basicstyle=\scriptsize\upshape\ttfamily]
$Q_1:$SELECT s.name, AVG(r.grade) as avg_grade
FROM Student s, Registration r 
WHERE s.name=r.name AND r.dept='CS'
GROUP BY s.name
HAVING COUNT(r.course)>=3
\end{lstlisting}
\vspace{-3mm}
\begin{lstlisting}[basicstyle=\scriptsize\upshape\ttfamily]
$Q_2:$SELECT s.name, AVG(r.grade) as avg_grade
FROM Student s, Registration r 
WHERE s.name=r.name
GROUP BY s.name
HAVING COUNT(r.course)>=3
\end{lstlisting}
\begin{table}[!htb]
\vspace{-3mm}
\scriptsize{
    \begin{minipage}{.4\linewidth}
	    \centering
	   	\caption*{Result of $Q_1(D)$:}
		\begin{tabular}{| l | l |}
			\hline
			name & avg\_grade \\ \hline
			Jesse & 90 \\
			\hline
		\end{tabular}
	\end{minipage}
    \begin{minipage}{.4\linewidth}
	    \centering
	   	\caption*{Result of $Q_2(D)$:}
		\begin{tabular}{| l | l |}
		\hline
		name & avg\_grade \\ \hline
		Mary & 90 \\
		Jesse & 90 \\
		\hline
		\end{tabular}
	\end{minipage}
	}
 \vspace{-3mm}
\end{table}
Again, $Q_2$ returns $(Mary, 90)$ that should not be in the correct result, because it does not select on departments. And we have to return all three courses Mary registered to make $(Mary, 90)$ still in the result of $Q_2$ but not in the result of $Q_1$. Therefore, when the selection predicate involves the comparison between count or sum with a large constant number, we have to return a large fraction of tuples in the test database instance in order to make the output tuple satisfy the selection predicate. 

\end{Example}

\subsection{Applying Provenance for Aggregates}


To address the first two challenge in Section~\ref{sec:challenges-agg} (the third challenge is discussed in Section~\ref{sec:enhance}), we consider applying provenance for aggregated queries by Amsterdamer et al.\cite{amsterdamer2011provenance}. Their approach annotates the provenance information with the individual values within tuple using commutative monoid (for aggregate) and commutative semirings (for annotation). The tuples in the input relations are regarded as symbolic variables and thus the aggregate values can be encoded as symbolic expressions. The selection predicates that involve aggregate values can be encoded as symbolic logical expressions. Then we can express $Q_1(D') \neq Q_2(D')$ using symbolic inequality expressions: assert that a group only exists in one of the query results, or the group exists in both query results but the aggregate values are different. Table~\ref{table:prov-selection-agg-2} shows the provenance of aggregate queries for Example~\ref{example:selection-agg-2}. For instance, $t_4 \otimes 100 +_{AVG} t_5 \otimes 75$ represents the AVG value of a group containing two tuples $t_4$ and $t_5$ in the original query result, and the value of the attribute in the AVG function of tuple $t_4$ if 100, and the value is 75 for $t_5$. If $t_4$ is removed from the input relations, then $t_4 \otimes 100$ will not contribute to the AVG value. Like the how-provenance, $(t_1 (t_4 + t_5)) \ (t_4 \otimes 1 +_{SUM} t_5 \otimes 1 \geq 3)$ indicates how the result tuple is derived from the input or intermediate tuples: $t_1 (t_4 + t_5)$ means that the group exists iff $t_1$ exists and one of $t_4$ and $t_5$ exists; $t_4 \otimes 1 +_{SUM} t_5 \otimes 1 \geq 3$ represents the selection criterion: the COUNT (a special case of SUM) value should be greater or equal to 3. Based on these provenance expressions, a counterexample for $Q_1$ and $Q_2$ w.r.t. tuple $(Mary, 90)$ can be given by solving the constraint $(prv_4 \oplus prv_1) \lor (val_4 \neq val_1)$, and we can iterate over all output tuples to find the smallest counterexample, instead of finding a global smallest witness of tuples in $Q_1(D) \setminus Q_2(D)$.

\begin{figure*}[htb]
\begin{minipage}{0.88\linewidth} \centering
{\scriptsize\upshape
\begin{tabular}{| l | l | l | l | l |}
	\hline
	$Q_1$ & & \\ \hline
	name & avg\_grade & & provenance & \\ \hline
	Mary & $t_4 \otimes 100 +_{AVG} t_5 \otimes 75$& $val_1$ &$ (t_1 (t_4 + t_5)) \ (t_4 \otimes 1 +_{SUM} t_5 \otimes 1 \geq 3) $& $prv_1$\\
	John & $t_7 \otimes 90 $& $val_2$ &$ (t_2 t_7) \ (t_7 \otimes 1 \geq 3) $& $prv_2$\\
	Jesse & $t_{9} \otimes 95 +_{AVG} t_{10} \otimes 90$ &$val_3$ & $((t_3 \ (t_{9} + t_{10})) \ (t_{9} \otimes 1 +_{SUM} t_{10} \otimes 1 +_{SUM} t_{11} \otimes 1 \geq 3)$& $prv_3$\\
	\hline
	$Q_2$ & & \\ \hline
	name & avg\_grade & & provenance &\\ \hline
	Mary & $t_4 \otimes 100 +_{AVG} t_5 \otimes 75 +_{AVG} t_6 \otimes 95$& $val_4$ &$ (t_1  (t_4 + t_5 + t_6)) \ (t_4 \otimes 1 +_{SUM} t_5 \otimes 1 +_{SUM} t_6 \otimes 1 \geq 3)$& $prv_4$\\
	John & $t_7 \otimes 90 +_{AVG} t_8 \otimes 88 $& $val_5$ &$ (t_2  (t_7 + t_8)) \ (t_7 \otimes 1 +_{SUM} t_8 \otimes 1 \geq 3)$& $prv_5$\\
	Jesse & $t_{9} \otimes 95 +_{AVG} t_{10} \otimes 90$ &$val_6$ & $(t_3  (t_{9} + t_{10} + t_{11})) \ (t_{9} \otimes 1 +_{SUM} t_{10} \otimes 1 +_{SUM} t_{11} \otimes 1 \geq 3) $& $prv_6$\\
	\hline
\end{tabular}
}
\captionof{table}{Provenance for Aggregate Queries in Example~\ref{example:selection-agg-2}}
\label{table:prov-selection-agg-2}
\end{minipage}
\vspace{-2.5mm}
\end{figure*}

\cut{By solving the constraints from provenance for aggregate queries, we can find the smallest counterexample instead of finding the smallest witness. }

\subsection{Optimizations}\label{sec:enhance}
\label{subsec:agg-enh}
The provenance-based approach can be optimized further.

\subsubsection{Parameterizing the Queries}

To address the third challenge, when the queries involve comparisons on aggregate values with constant numbers, we modify the definition of our problem by parameterizing the queries. We replace the constants in selection predicates with symbolic variables when passing the provenance expressions to the solver. Then we are expected to get a smaller counterexample with different constant values in the selection predicates, compared to the one under the original parameter settings.


\begin{Definition}[Smallest Parameterized Counterexample Problem]
Given two parameterized queries $Q_1$ and $Q_2$, and a parameter setting $\lambda$ and a database instance $D$, where $Q_1(\lambda, D) \neq Q_2(\lambda, D)$, the smallest parameterized counterexample problem (SPCP) is to find a parameter setting $\lambda'$ and a subinstance $D'$ of $D$, such that $Q_1(\lambda', D') \neq Q_2(\lambda', D')$, and the total number of tuples in $D'$ is minimized.
\end{Definition}

\begin{Example}[Smallest Parameterized Counterexample]

\label{example:selection-agg-3}
Here we show an example of parameterized queries based on Example~\ref{example:selection-agg-2} by making the number of CS courses in the queries a parameter.
\begin{lstlisting}[basicstyle=\scriptsize\upshape\ttfamily]
$Q_1:$ SELECT s.name, AVG(r.grade) as avg_grade
FROM Student s, Registration r 
WHERE s.name=r.name AND r.dept='CS'
GROUP BY s.name
HAVING COUNT(r.course)>= $@num_{CS}$
\end{lstlisting}
\vspace{-3mm}
\begin{lstlisting}[basicstyle=\scriptsize\upshape\ttfamily]
$Q_2:$ SELECT s.name, AVG(r.grade) as avg_grade
FROM Student s, Registration r 
WHERE s.name=r.name
GROUP BY s.name
HAVING COUNT(r.course)>=$@num_{CS}$
\end{lstlisting}
\vspace{-1mm}
These two queries return:
\vspace{-5mm}
\begin{table}[!htb]
\scriptsize{
{\upshape
    \begin{minipage}{.48\linewidth}
	    \centering
	   	\caption*{$Q_1(num_{CS} = 3, D)$:}
	   	{\scriptsize
		\begin{tabular}{| l | l |}
			\hline
			name & avg\_grade \\ \hline
			Jesse & 90 \\
			\hline
		\end{tabular}
		}
	\end{minipage}
	\begin{minipage}{.48\linewidth}
	    \centering
	   	\caption*{$Q_2(num_{CS} = 3, D)$:}
	   	{\scriptsize
		\begin{tabular}{| l | l |}
		\hline
		name & avg\_grade \\ \hline
		Mary & 90 \\
		Jesse & 90 \\
		\hline
		\end{tabular}
		}
	\end{minipage}
	}
}
 \vspace{-3mm}
\end{table}

By using a different parameter setting, the size of counterexample can be reduced. When $@num_{CS}=3$, the smallest counterexample $C$ is $t_1, t_4, t_5, t_6$. But if $@num_{CS} = 1$, we only need to return $t_1, t_6$.
\end{Example}

Below is an example illustrating how to encode the provenance for aggregate queries to SMT formulas.

\begin{lstlisting}[language=smtlib2,style=smtlib2,caption={SMT-LIB Input for Example~\ref{example:selection-agg-3}},label={sample-smtlib2-selection-agg-3},captionpos=b]
(declare-const t1 Bool)
...
(declare-const t11 Bool)
(declare-const num_CS Int)
(define-fun b2i ((x Bool)) Int (ite x 1 0))
(assert 
	(or 
		(distinct 
			(and 
				(and t1 (or t4 t5)) 
				(>= (+ (b2i t4) (b2i t5)) num_CS))
			(and 
				(and t1 (or t4 t5 t6)) 
				(>= (+ (b2i t4) (b2i t5) (b2i t6)) num_CS))) 
		(not 
			(= 
				(\ (+ (* (b2i t4) 100) (* (b2i t5) 75)) (+ (b2i t4) (b2i t5))) 
				(\ (+ (* (b2i t4) 100) (* (b2i t5) 75) (* (b2i t6) 95)) (+ (b2i t4) (b2i t5) (b2i t6)))
))))
(minimize (+ (b2i t1) (b2i t2) ... (b2i t11)))
\end{lstlisting}
\subsubsection{A Heuristic Approach}

The provenance-based solution may not scale very well when a group contains too many tuples and thus the SMT formulas involve too many variables, even if we choose the group with the least number of tuples. Assume that the aggregate functions and attributes are the same in two queries, to reduce the number of variables in SMT formulas, we decide to look into the different tuples between two groups.
E.g., for simplicity, 
assume that both $Q_1$ and $Q_2$ are in the form of $\aggregation{G_1}{agg_1(A_1), agg_2(A_2), ..., agg_k(A_k)}(Q'_1(D))$ and $\aggregation{G_2}{agg_1(A_1), agg_2(A_2), ..., agg_k(A_k)}(Q'_2(D))$ (the aggregations are done at last), one of the following two cases must be true:
(i) the group in $Q_1(D)$ that generates $t$ does not exist in $Q_2(D)$ (grouping attributes $G_1$ may or may not be equal to $G_2$);  (ii) the group in $Q_1$ that generates $t$ also exists in $Q_2$, but one of the aggregate values are different. In either case, we can directly compare the result of $Q'_1(D)$ and $Q'_2(D)$ and find at least one tuple that exists in only one of them. The following example illustrates how this method works. Note that $Q'_1$ and $Q'_2$ can include nested aggregate queries, as long as there are no aggregate values in their schema --- aggregate values can be involved in selections or joins.

\begin{Example}[Heuristic Approach on Example~\ref{example:agg1}]
\vspace{-2mm}
\begin{lstlisting}[basicstyle=\scriptsize\upshape\ttfamily]
$Q_1':$
SELECT s.name, r.grade
FROM Student s, Registration r 
WHERE s.name=r.name AND r.dept='CS'
\end{lstlisting}
\vspace{-1mm}
\begin{lstlisting}[basicstyle=\scriptsize\upshape\ttfamily]
$Q_2':$SELECT s.name, r.grade
FROM Student s, Registration r 
WHERE s.name=r.name
\end{lstlisting}
\begin{table}[!htb]
\vspace{-3mm}
{\scriptsize\upshape
    \begin{minipage}{.45\linewidth}
	    \centering
	   	\caption*{$Q'_1(D)$:}
		\begin{tabular}{|p{2em}|p{2em}||p{2em}|p{2em}|}
		\hline
		name & grade & name & grade\\ \hline
		Mary & 100 & Jesse & 95 \\
		Mary & 75 & Jesse & 90 \\
		John & 90 & Jesse & 85 \\
		\hline
		\end{tabular}
	\end{minipage}
	\begin{minipage}{.45\linewidth}
	    \centering
	   	\caption*{$Q'_2(D)$:}
		\begin{tabular}{|p{2em}|p{2em}||p{2em}|p{2em}|}
		\hline
		name & grade & name & grade \\ \hline
		Mary & 100 & John & 88 \\
		Mary & 95 & Jesse & 95 \\
		Mary & 75 & Jesse & 90 \\
		John & 90 & Jesse & 85 \\
		\hline
		\end{tabular}
	\end{minipage}
}
\vspace{-2mm}
\end{table}

$Q_2'$ does not select on departments so it returns some additional tuples comparing to $Q_1'$: $(Mary, 95)$ and $(John, 88)$. And now we can apply the method for SPJUD queries in the previous section and then return either \cut{Mary's ECON 208D record or John's ECON 208D record} $\{t_1, t_6\}$ or $\{t_2, t_8\}$ as the counterexample --- they can explain why the aggregate value on Mary or John are different in $Q_1$ and $Q_2$.
\label{example:agg-heuristic}
\end{Example}

\begin{algorithm}[t]\caption{$Agg_{opt}$: The heuristic algorithm for aggregate queries}\label{alg:agg-heu}
\small
\begin{codebox}
\Procname{$\proc{Smallest-Counterexample-Aggregate-Heu}(Q_1, Q_2, D, \Lambda)$}
\li $\Lambda = \{\lambda_1, ... \}$ is the original parameter setting
\li $Q_1 = \selection_{agg_1(A_1)\ op \ \lambda_1}\aggregation{G_1}{agg_1(A_1), agg_2(A_2), ..., agg_k(A_k)}(Q'_1(D))$,
\li $Q_2 = \selection_{agg_1(A_1)\ op \ \lambda_1}\aggregation{G_2}{agg_1(A_1), agg_2(A_2), ..., agg_k(A_k)}(Q'_2(D))$
\li \Repeat
\li Pick one tuple $t$ in the result of $Q'_1(D) \setminus Q'_2(D)$, 
\li $A'_1 ... A'_k \gets$ the attributes of $t$
\li $Q' \gets \sigma_{A'_1 = t.A'_1, A'_2 = t.A'_2, ..., A'_k = t.A'_k} (Q'_1-Q'_2)$
\li $\phi \gets $ the prv. for agg. queries of tuple $t$ w.r.t. $Q'(D)$
\li $obj \gets$ \#true\_values in $\phi$
\li $\eta \gets$ OptSMT\_Solver($\phi$, $obj$)
\li $D_{\eta} \gets \{t' \mid \eta(t') \ is \ true \}$
\li Set $\Lambda'$ according to values in $D_{\eta}$
\li \Until {$Q_1(D_{\eta}, \Lambda') \neq Q_2(D_{\eta}, \Lambda')$}
\li Return $D_{\eta}$
\end{codebox}
\vspace{-1mm}
\end{algorithm}
When the queries involve comparisons on aggregate values at the top of the query tree, e.g., $\selection_{agg_1(A_1)\ op \ const}$ \\$\aggregation{G^1}{agg_1(A_1), agg_2(A_2), ..., agg_k(A_k)}(Q'_1(D))$, we can also apply the heuristic approach by parameterizing the queries and directly looking into $Q'_1$ and $Q'_2$
\cut{: (1). if the aggregate function in the selection predicate is MIN, MAX, or AVG, rewrite to $\selection_{A_1\ op \ const} (Q'_1(D))$ and  $\selection_{A_1\ op \ const} (Q'_2(D))$; (2). if the aggregate function is COUNT or SUM, we drop the selection and rewrite the queries to $Q'_1(D)$ and $Q'_2(D)$}.
After finding the smallest counterexample $C$ for $Q'_1$ and $Q'_2$, the next step is to make sure $Q_1(C) \neq Q_2(C)$ otherwise it fails to distinguish the original queries. On one hand, if aggregate values are involved in the selection predicate at the top of the query tree, we  parameterize the original queries and set a reasonable number such that the results of at least one of $Q'_1(C)$ and $Q'_2(C)$ will pass the selection. For COUNT we set the parameter in the predicate to be 1 or 0 if the operator is `=' or `>', while for SUM we set the parameter to be the maximum value or the minimum value of the attribute in the aggregate function. And it is similar for MAX, MIN, and AVG. On the other hand, if both $Q_1(C)$ and $Q_2(C)$ are not empty, their results may happen to be the same since we do not add any constraints on the aggregate values --- the only guarantee is $Q'_1(C) \neq Q'_2(C)$. Therefore we have to evaluate the queries on the counterexample we find, and if $Q_1(C) = Q_2(C)$, we re-run the SMT Solver on the same formulas but asking the solver to return a different model until we get a satisfying counterexample.

\section{Implementation}
\label{sec:impl}

Provenance has been extensively studied in the database community, not only theoretically \cite{green2007provenance, amsterdamer2011provenance, buneman2001and}, but also there are systems that can capture different forms of provenance \cite{karvounarakis2010querying, glavic2009perm, arab2014generic, green2007update, psallidas2018smoke, senellart2018provsql}. However, to the best of our knowledge, there are no systems available that support how-provenance for general SPJUD and aggregate queries. Since building a comprehensive system to efficiently capture provenance is not the goal of this paper, for simplicity, we implemented our system, called \OurSys{}, in Python 3.6 with \aRADB{} \cite{anonymous2017rai}. This interpreter translates relational algebra queries into SQL common table expression (CTE) queries, and each relational algebra operator is translated into a SQL subquery. \OurSys{} has a web UI built using HTML, CSS, and JavaScript. It runs on Ubuntu 16.04 and uses Microsoft SQLServer 2017 as the underlying DBMS.

First, \OurSys{} takes two queries in Relational Algebra as input. Then \theRADB{} interprets them and generates two SQL queries consisting of multiple subqueries. Next, it rewrites each subquery by adding one additional column of provenance expression. For aggregate queries, all columns of aggregate values are also rewritten to symbolic expressions. These expressions are stored as strings in the SMT-LIB format. For each input relation, we added one additional `\emph{prv}' column of tuple identifiers. The rewriting rules are listed below by the relational algebra operator:


\mypar{Select, Project, Union} For selection/projection/union, we directly select the \emph{prv} column from the input relation. If the selection predicate involves aggregate values (i.e., HAVING), we construct a symbolic logical expression with the operands, and take the conjunction of the \emph{prv} column and the symbolic logical expression. Duplicates from projection/union will be considered in de-duplicate discussed below.

\mypar{Join} We take the {\tt AND} ($\wedge$) of the \emph{prv} column of the two joining tuples.

\mypar{Difference} Remember that for difference operator, there are two cases while evaluating $R - S$: (1) $t \in R, t \in S$. (2) $t \in R, t \notin S$. In the first case, the query is transformed into a join query where the join predicate is that the tuple in $R$ should equal to the tuple in $S$ (excluding the \emph{prv} column); In the second case, we add a `{\tt NOT EXISTS}' subquery to find those tuples in $R$ but not in $S$, and the \emph{prv} column is the same as those in $R$; then we take a union of these two cases.

\begin{lstlisting}[basicstyle=\scriptsize\upshape\ttfamily]
SELECT R.A, R.B FROM R EXCEPT SELECT S.A, S.B FROM S
\end{lstlisting}
\vspace{-2mm}
is rewritten to:
\vspace{-2mm}
\begin{lstlisting}[basicstyle=\scriptsize\upshape\ttfamily]
(SELECT R.A, R.B, 
     '(and ' || R.prv || ' (not ' || S.prv || '))'
FROM R, S WHERE R.A = S.A AND R.B = S.B) UNION
(SELECT R.A, R.B, R.prv
FROM R WHERE NOT EXISTS(
     SELECT S.A, S.B
     FROM S
     WHERE S.A = R.A AND S.B = R.B))
\end{lstlisting}
\vspace{-3mm}


\mypar{De-duplicate} Duplicate tuples may arise from projection or union. We add a group by clause that contains all columns in the select clause except the \emph{prv} column. The \emph{prv} column is computed using `string\_agg' function: 
\vspace{-1mm}
\begin{lstlisting}[basicstyle=\scriptsize\upshape\ttfamily]
'(or ' || string_agg(R.prv, ' ') || ')'
\end{lstlisting}


Once the queries are rewritten, \OurSys{} applies the algorithms in Section~\ref{sec:general-algo} and Section~\ref{sec:algo_agg} according to their query classes, to generate SMT constraints. Then, \OurSys{} passes the constraints to the Z3 SMT Solver (an efficient optimizing SMT Solver by Microsoft Research)\cite{de2008z3,bjorner2015nuz} 4.7.1, and sets ``\emph{minimizing the number of variables set to true}'' as the objective function. Finally, the satisfying model returned by the Solver represents the counterexample, and the counterexample is shown on the web UI with the query results of two input queries over this counterexample.

\section{Experiments}
\label{sec:experiments}

We present experiments to evaluate our algorithms in this section. The
input queries used in our initial experiment for SPJUD queries were
collected from student submissions to a relational algebra assignment
in an undergraduate database course in a US university in Fall 2017;
therefore, the wrong queries were ``real,'' although test database
instances are synthetic. In the experiment for aggregate queries, we
use the TPC-H benchmark\cite{council2008tpc}. We generate tables at
scale 1, and manually translated several TPC-H queries into Relational
Algebra and created some wrong queries ourselves.
The system runs locally on a 64-bit Ubuntu 16.04 LTS with 3.60 GHz Intel Core i7-4790 CPU and 16 GB 1600 MHz DDR3 RAM.

\vspace{-2mm}
\subsection{Real World SPJUD Queries}

In this subsection we evaluate the efficacy of our algorithms in
Section~\ref{sec:general-algo} for SPJUD queries on the course
dataset. The dataset comes from one relational algebra assignment in
Fall 2017, which asked students to write SPJUD queries using
\theRADB{}. It includes 8 questions and 141 students in total, and the
queries are evaluated over the test instances we generated. The test
instance may not be able to differentiate all incorrect queries,
although there are more incorrect queries discovered when the test
instance gets larger (see Table~\ref{table:wrong-query-stats}).  Some
questions involve very complex queries to find tuples satisfying
conditions with universal quantification or uniqueness quantification
requiring multiple uses of difference (see
Section~\ref{sec:user-study} for concrete examples), and solicited
some extremely complex student solutions with scores of operators; we
are not aware of any directly related work that is able to handle this
level of query complexity. \cut{Only in our experiments on the largest
test database instance (of 100,000 tuples), we had to drop two overly
complicated student queries that involved massive cross products.} We had to drop two overly complicated student queries that involved massive cross products.
\begin{table}[htb]
\scriptsize{
  \centering
  \begin{tabular}{|>{\centering\arraybackslash}p{5em}|>{\centering\arraybackslash}p{8em}|>{\centering\arraybackslash}p{8em}|}
  \hline  
  \# of Tuples in DB & \# of incorrect queries & \# of students with incorrect queries \\ \hline
  1,000 & 111 & 76  \\
  4,000 & 167 & 87 \\
  10,000 & 168 & 88 \\
  40,000 & 169 & 88 \\
  100,000 & 170 & 88 \\
  \hline
  \end{tabular}
  \caption{$|D|$ vs. \# of wrong queries discovered}
  \label{table:wrong-query-stats}
}
  \vspace{-2mm}
\end{table}
\vspace{-1mm}

\mypar{\SCP \ vs. \SWP} As discussed in Section~\ref{sec:prelim}, a poly-time solution for \SWP\ also gives a poly-time solution for \SCP\ if we consider data complexity, since we can iterate over all tuples $t$ in $Q_1(D)\setminus Q_2(D)$ to find the global optimal solution. The number of output tuples is polynomial in $|D|$, but can be exponential in query size (\eg, when we join $k$ tables that form a cross product), and therefore it does not necessarily give a poly-time solution in terms of combined complexity. However, the standard practice is to consider data complexity, since the size of the query is expected to be a small constant. For practical purposes, even polynomial combined complexity may not give interactive performance. Here we experimented on the algorithms for SPJUD queries in Section~\ref{sec:general-algo} to compare \SCP\ and \SWP\ in practice: The \spjudBasic{} algorithm using Z3 SMT optimizer instead of a SAT solver and the \spjudOpt{} algorithm. See Table~\ref{table:scp-vs-swp}. Surprisingly, all smallest witnesses returned by \spjudOpt{} are of the same size as the smallest counterexamples returned by \spjudBasic{}, i.e., \spjudBasic{} reaches the global optimum on the first output tuple. This may be a coincidence, but for 168 of 170 wrong queries we discovered, the size of the smallest witnesses of all output tuples are the same. The result indicates that in most cases, we can use \spjudOpt{}~for better performance (6.9x faster) with only a small probability of not reaching the global minimum. Given this result, in the rest of this section, we will only experiment on \SWP.

\begin{table}[htb]
\scriptsize
  \centering
  \begin{tabular}{|c|>{\centering\arraybackslash}p{7em}|>{\centering\arraybackslash}p{7em}|}
  \hline  
  & Mean Runtime (sec.) & Mean Size of Counterexample\\
  \hline
   \SCP --- \spjudBasic & 26.29 & 3.52\\
   \SWP --- \spjudOpt{} & 3.80 & 3.52 \\
  \hline
  \end{tabular}
  \caption{\SCP \ vs. \SWP, \# tuples = 100k}
  \label{table:scp-vs-swp}
  \vspace{-2mm}
\end{table}

\mypar{Size of the data vs. time} We vary the number of tuples in the input relations of 1,000: 4,000: 10,000: 40,000: 100,000. See figure \ref{fig:running-time-dbsize}: \textbf{raw} is for evaluating queries $Q_1-Q_2$, the difference of students' query and the standard query; \textbf{prov-all} is for evaluating rewritten queries $Q_1-Q_2$ that also store provenance; \textbf{prov-sp} is for provenance queries with selection on one tuple; \textbf{solver-naive-128} is for finding the smallest witness with an SAT solver that tries at most 128 different models; \textbf{solver-opt} is for finding the smallest witness of the first result tuple with Z3 SMT optimizer; \textbf{solver-opt-all} is for finding the smallest witness of all result tuples with Z3 SMT optimizer. The running time of rewritten provenance queries with selection pushdown is much faster than the raw queries (29x when $|D| = 100K$) and the provenance query without selection on tuples (42x when $|D| = 100K$). This is what we expect: computing provenance expression will cause huge overheads, but only for one single tuple is definitely affordable.



\mypar{Query complexity vs. time} Figure~\ref{fig:query-complexity-vs-time} shows the running time of each component of \spjudOpt{}\ vs. different metrics of the query complexity (number of operators, number of differences, and height of the query tree). The running time increases roughly as the complexity of queries increases. Note that when height of the query or the number of operators in the query reaches the maximum, the provenance query dominates the running time, however, in most cases, evaluating the raw CTE SQL query is the slowest part.

\begin{figure*}\centering
  \begin{minipage}{0.88\linewidth}
  \begin{minipage}{0.295\linewidth}
	\centering
  \includegraphics[width=\linewidth]{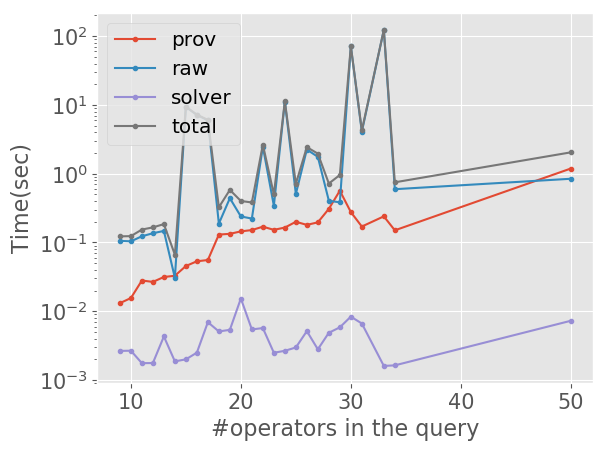}
  \end{minipage}
    \hfill
  \begin{minipage}{0.295\linewidth}
	\centering
  \includegraphics[width=\linewidth]{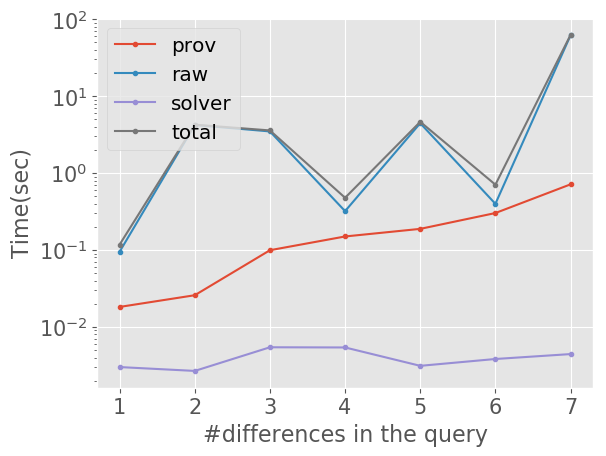}
  \end{minipage}
  \hfill
  \begin{minipage}{0.295\linewidth}
	\centering
  \includegraphics[width=\linewidth]{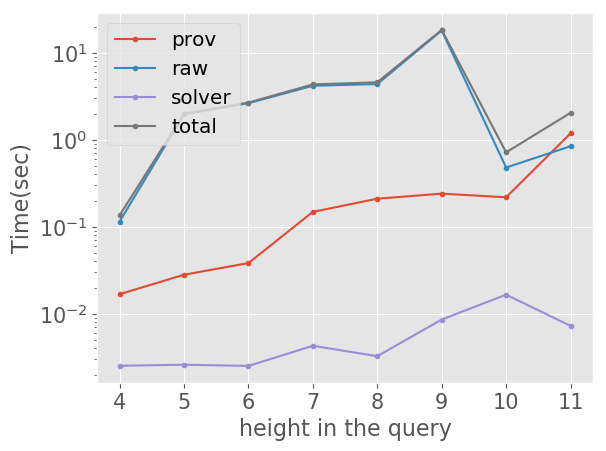}
  \end{minipage}
  \caption{Query complexity vs. time, algorithm=\spjudOpt{}, \#tuples=100k\cut{40k}; \textbf{raw} is for evaluating queries $Q_1-Q_2$; \textbf{prov-sp} is for provenance queries with selection on one tuple; \textbf{solver-opt} is for finding the smallest witness with Z3 SMT optimizer; \textbf{total} is the total running time of \spjudOpt.}
  \label{fig:query-complexity-vs-time}
  \end{minipage}\\[-1mm]
  \begin{minipage}{0.26\linewidth}
  	\centering
  \includegraphics[width=\linewidth]{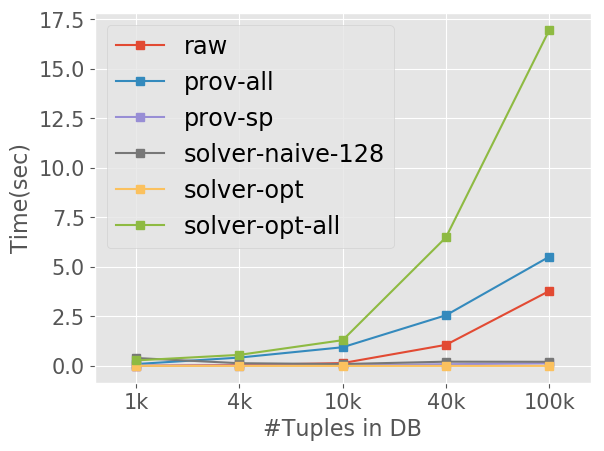}
	\caption{Average Running Time of Each Component.}
	\label{fig:running-time-dbsize}
  \end{minipage}
  \hspace{0.03\linewidth}
  \begin{minipage}{0.63\linewidth}
  \begin{minipage}{0.413\linewidth}
	\centering
  \includegraphics[width=\linewidth]{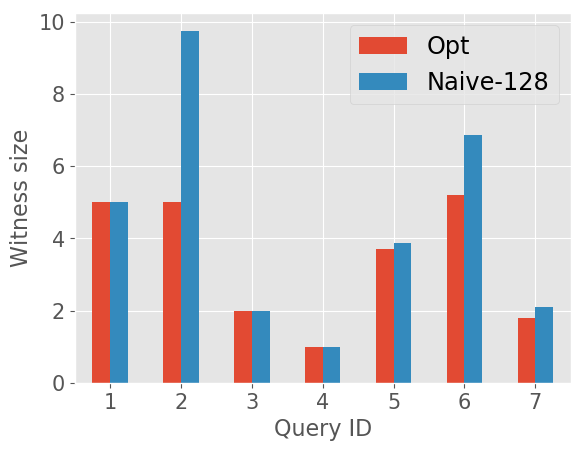}
  \end{minipage}
  \hfil
  \begin{minipage}{0.413\linewidth}
	\centering
  \includegraphics[width=\linewidth]{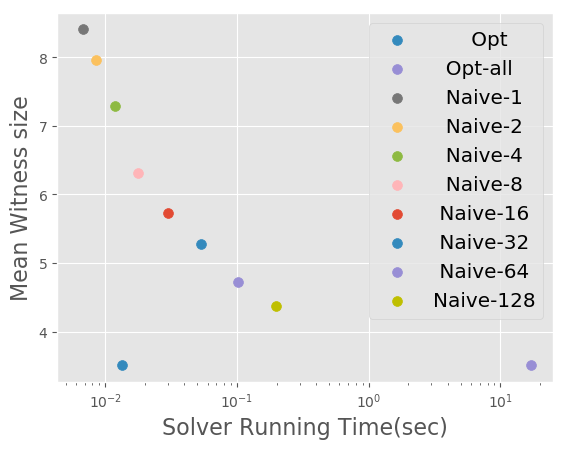}
  \end{minipage}
  \caption{(a) Witness size vs. solver strategies and Queries, \#tuples=100K; (b) Witness size vs. solver strategies and runtime, \#tuples = 100K}
  \label{fig:witness-size-solver}
  \end{minipage}
    \vspace{-2mm}

 \end{figure*}
\cut{
\mypar{Solver strategy vs. time} We experiment different constraint-solving strategies: \textbf{Naive-128} is to use the SMT solver to get satisfying models to the Boolean formula of how-provenance, until there are no more satisfying models or it finishes enumerating M different models (we chose M=128); \textbf{Opt} is to use an SMT optimizer to directly find the model with least number of variables set to true. In Figure~\ref{fig:running-time-dbsize}, we can see that \textbf{Opt} is much faster than \textbf{Basic} for smaller relations. 
}
\cut{
\begin{figure}
\centering
\includegraphics[width=\linewidth]{imgs/new_plots/witness_size_against_solver_time.png}
\caption{Solver Runtime Against Witness Size}
\label{fig:witness-size-time}
\end{figure}
}

\mypar{Solver strategy vs. witness size} Since our goal is to find the
smallest witness, the metric for evaluating the quality of the
witnesses as explanations is the size. We experiment different
constraint-solving strategies: \textbf{Naive-*} is to use the Z3 SMT
solver to get satisfying models to the Boolean formula of
how-provenance, until there are no more satisfying models or it
finishes enumerating M different models (we chose M=128); \textbf{Opt}
is to use Z3 SMT optimizer to directly find the model with least
number of variables set to true. \textbf{Naive-*} is not satisfying
because there is no guarantee on the model size it
finds. Figure~\ref{fig:witness-size-solver} summarizes the
results. Since the SAT solver used by \textbf{Naive-*} can return an
arbitrary model every time, we repeat each experiment with
\textbf{Naive-*} 10 times and report the average minimum witness size
found among the 10 repetitions. \textbf{Opt} always return a smaller
witness compared to \textbf{Naive-*}, while the runtime overhead
compared to even \textbf{Naive-1} is negligible. Of course,
performance of these approaches heavily depends on the solver
implementation; a comprehensive evaluation would be beyond the scope
of this paper.  Here, we simply observe that our implementation of
\textbf{Opt} provides good performance and solution quality in
practice, as it cannot be easily beaten by simply enumerating a number
of models.

\vspace{-2mm}
\subsection{Synthetic Aggregate Queries}
We experimented on the TPC-H benchmark database generated at scale 1 on queries Q4, Q16, Q18, Q21, and a modified Q21-S with an additional selection on aggregate value at the top of the query tree. We choose these queries because they do not involve arithmetic operations on aggregate functions. First we dropped the ORDER BY operator and rewrote these queries using \theRADB{}, then we experimented both the provenance for aggregate queries approach (\emph{\aggBasic{}}) and the heuristic approach (\emph{\aggHeu{}})  discussed in Section~\ref{sec:algo_agg}. We also experimented provenance for aggregate queries approach with parameterization (\emph{\aggParam}) on Q18 (it has a selection predicate with aggregate value). We intentionally made two wrong queries for each query, of which the errors include different selection conditions, incorrect use of difference, and incorrect position of projection. These are common errors in the students queries from the previous experiment.


\begin{figure*}
\scriptsize
  \begin{minipage}{0.98\linewidth}
  \begin{minipage}{0.62\linewidth}
  \centering
  \begin{tabular}{|>{\centering\arraybackslash}p{4em}|>{\centering\arraybackslash}p{5em}|>{\centering\arraybackslash}p{5em}|>{\centering\arraybackslash}p{5em}|>{\centering\arraybackslash}p{5em}|>{\centering\arraybackslash}p{5em}|>{\centering\arraybackslash}p{5em}|}
  \hline  
  \multirow{2}{*}{Query} & \multicolumn{3}{c|}{\aggBasic{}} & \multicolumn{3}{c|}{\aggHeu{}}\\ \cline{2-7}
  & Raw Query Eval. Time & Prov. Query Eval. Time & Solver Runtime & Raw Query Eval. Time & Prov. Query Eval. Time & Solver Runtime \\
  \hline
  Q4 &  3.6036 & 4.0403  & timeout & 2.1382 & 0.0029 & 0.0151\\
  Q16 & 0.8676 &  0.1349 & 0.2471 &  0.7618 & 0.1084 & 0.0022\\
  Q18  & 6.8751 & 0.0086 &  0.0134 & 14.2513 & 0.0130 &  0.0039\\
  Q21& 21.5184 & 2.6205 & 31.1106  & 21.8072 & 0.0577 & 0.0066 \\
  Q21-S& 21.5408 & 2.8034 & 155.6828  & 22.1634 & 0.0524 & 0.0061 \\
  \hline
  \end{tabular}
  \caption{Computation time (sec.), for the TPC-H benchmark, scale=1, timeout = not finishing after 2 hours}
  \label{table:running-time-agg}
\end{minipage}
\hfil
\begin{minipage}{0.28\linewidth}
  \centering
  \begin{tabular}{|c|>{\centering\arraybackslash}p{6em}|>{\centering\arraybackslash}p{6em}|}
  \hline  
  & Solver Runtime (sec.) &Size of Counterexample \\
  \hline
   \aggBasic{} & 0.0134 & 25.3\\
   \aggParam{} & 0.0210  & 7.5 \\
  \hline
  \end{tabular}
  \caption{Effectiveness of algorithm with parameterization, on TPC-H Q18, scale-factor = 1}
  \label{table:agg-parameterization}
\end{minipage}
\end{minipage}
\vspace{-3.5mm}
\end{figure*}

\cut{
\begin{table*}[htb]
\scriptsize{
  \centering
  \begin{tabular}{|>{\centering\arraybackslash}p{4em}|>{\centering\arraybackslash}p{5em}|>{\centering\arraybackslash}p{5em}|>{\centering\arraybackslash}p{5em}|>{\centering\arraybackslash}p{5em}|>{\centering\arraybackslash}p{5em}|>{\centering\arraybackslash}p{5em}|}
  \hline  
  \multirow{2}{*}{Query} & \multicolumn{3}{c|}{\aggBasic{}} & \multicolumn{3}{c|}{\aggHeu{}}\\ \cline{2-7}
  & Raw Query Eval. Time & Prov. Query Eval. Time & Solver Runtime & Raw Query Eval. Time & Prov. Query Eval. Time & Solver Runtime \\
  \hline
  Q4 &  3.6036 & 4.0403  & timeout & 2.1382 & 0.0029 & 0.0151\\
  Q16 & 0.8676 &  0.1349 & 0.2471 &  0.7618 & 0.1084 & 0.0022\\
  Q18  & 6.8751 & 0.0086 &  0.0134 & 14.2513 & 0.0130 &  0.0039\\
  Q21& 21.5184 & 2.6205 & 31.1106  & 21.8072 & 0.0577 & 0.0066 \\
  Q21-S& 21.5408 & 2.8034 & 155.6828  & 22.1634 & 0.0524 & 0.0061 \\
  \hline
  \end{tabular}
  \caption{Computation time (sec.), for the TPC-H benchmark, scale=1, timeout = not finishing after 2 hours}
  \label{table:running-time-agg}
  \vspace{-2mm}
}
\end{table*}}
Figure~\ref{table:running-time-agg} includes the runtime of our algorithms to find the smallest counterexample for each TPC-H query we experiment. We present a breakdown of the execution time of our solutions: raw query evaluation time, provenance query evaluation time, SMT-solver running time. We find that the heuristic algorithm performs well for queries where the aggregation operators are at the top of the query tree. While the performance of the provenance for aggregate query algorithm decreases as the database size increases, and is significantly affected by the number of tuples in the group (The SMT solver does not scale well).

\cut{
\begin{lstlisting}[caption={TPC-H Q18, the parameter in the having predicate was set to 200, and the order by operator was dropped},label={tpch-q18}]
SELECT c_name, c_custkey, o_orderkey, 
  o_orderdate, o_totalprice, SUM(l_quantity)
FROM customer, orders, lineitem
WHERE
    o_orderkey IN (
        SELECT l_orderkey
        FROM lineitem
        GROUP BY l_orderkey 
        HAVING SUM(l_quantity) > ':1'
    )
    AND c_custkey = o_custkey 
    AND o_orderkey = l_orderkey
GROUP BY c_name, c_custkey, o_orderkey, 
o_orderdate, o_totalprice
ORDER BY o_totalprice DESC, o_orderdate;
\end{lstlisting}
}
For Q18, since it involves an aggregate predicate, we experiment the effectiveness of the algorithm with parameterization. Figure~\ref{table:agg-parameterization} shows the solver runtime and the size of the counterexample of the provenance for aggregate query algorithm with/without parameterization. The size of the counterexample is reduced by 70\% while the runtime only increases from 0.0134 seconds to 0.0210 seconds.
\cut{
\begin{table}[htb]
\scriptsize{
  \centering
  \begin{tabular}{|c|>{\centering\arraybackslash}p{6em}|>{\centering\arraybackslash}p{6em}|}
  \hline  
  & Solver Runtime (sec.) &Size of Counterexample \\
  \hline
   \aggBasic{} & 0.0134 & 25.3\\
   \aggParam{} & 0.0210  & 7.5 \\
  \hline
  \end{tabular}
  \caption{Effectiveness of algorithm with parameterization, on TPC-H Q18, scale-factor = 1}
  \label{table:agg-parameterization}
    \vspace{-2mm}}
\end{table}
}





\vspace{-2mm}
\section{User Study}
\label{sec:user-study}

Since one motivation of our work is to provide small examples as
explanations of why queries are incorrect, we built our \OurSys{} as a
web-based teaching tool and deployed in an undergraduate database
class in a US university in Fall 2018 with about 170 students.
For one homework assignment, students needed to write relational
algebra queries to answer 10 questions against a database of six
tables about bars, beers, drinkers, and their relationships.  The
difficulties of these 10 problems range from simple to very difficult.
The students had a small sample database instance to try their queries
on.  Their submissions were tested by an auto-grader against a large,
hidden database instance designed to exercise more corner cases and
catch more errors; if these answers differed from those returned by
the correct queries (also hidden), the students would see the failed
tests with some addition information about the error (but not the
hidden database instance or the correct queries).  The final
submissions were then graded manually informed
by the auto-grader results; partial credits were given.  For the
purpose of this user study, we normalize the student score for each
question to $[0, 100]$.

We did not wish to create unfair advantages for or undue burdens on
students with our user study.  This consideration constrained our user
study design.  For example, we ruled out the option of dividing
students into groups where only some of them benefit from \OurSys{};
we also ruled out creating additional homework problems without
counting them towards the course grades.  Therefore, we made the use
of \OurSys{} completely optional (and with no extra incentives other
than the help \OurSys{} offers itself).  \OurSys{} was given the
correct queries and the same database instance used by the auto-grader
for testing.  If a student query returned an incorrect result,
\OurSys{} would show a small database instance (a subset of the hidden
one), together with the results of the incorrect query and the hidden
correct query on this small instance.  We made \OurSys{} available for
only 5 out of the 10 problems.  Leaving some problems out allowed us
to study the same student's performance on different problems might be
influence by the use of \OurSys{}.  The 5 problems were chosen to
cover the entire range of difficulties:
\cut{
(b) Find the names of drinkers who frequent any bar serving
  Corona.
(d) Find the names of drinkers who frequent both James Joyce
  Pub and Satisfaction.
(e) Find the names of bars frequented by either Ben or Dan, but
  not both.
(g) For each bar, find the drinker who frequents it the
  greatest number of times.
(h) Find names of all drinkers who frequent only those bars
  that serve some beers they like.
  }

\begin{itemize}\itshape
 \vspace{-1mm}
\item[(b)] Find drinkers who frequent any bar serving
  Corona.
\item[(d)] Find drinkers who frequent both JJ Pub and Satisfaction.
\item[(e)] Find bars frequented by either Ben or Dan, but
  not both.
\item[(g)] For each bar, find the drinker who frequents it the
  greatest number of times.
\item[(h)] Find all drinkers who frequent only those bars
  that serve some beers they like.
 \vspace{-1.5mm}
\end{itemize}
Students must use basic relational algebra; in particular, they were
not allowed to use aggregation.  Problems (g) and (i) are more
challenging than others: (g) involves non-trivial uses of self-join
and difference; (i) involves two uses of difference.

We released \OurSys{} a week in advance of the homework due date.  We
collected usage patterns on \OurSys{}, as well as how students
eventually scored on the homework problems.  Ideally, we wanted to
answer the following questions: i)~Did students who used \OurSys{} do
better than those who didn't? ii)~For students who used \OurSys{}, how
did they do on problems with and without \OurSys{}'s help?  We should
note upfront that we expected no simple answers to these questions, as
scores could be impacted by a variety of factors, including the inherent
difficulty of a question itself, individual students' abilities and
motivation, as well as the learning effect (where one gets better at
writing queries in general after more exercises).  Therefore, to
supplement quantitative analysis of usage patterns and scores, we also
gave an optional, anonymous questionnaire to all
students after the homework due date.

\begin{figure*}[!h]
\scriptsize
\begin{minipage}{0.98\linewidth}\centering
\begin{minipage}{0.6\linewidth}\centering
\begin{minipage}{1.0\linewidth}\centering
	\begin{tabular}{|>{\centering\arraybackslash}p{4em}|>{\centering\arraybackslash}p{3em}|>{\centering\arraybackslash}p{7em}|>{\centering\arraybackslash}p{4em}|>{\centering\arraybackslash}p{7em}|}
	\hline	
	\multirow{2}{*}{Problem} & \multicolumn{2}{c|}{\# of users} & \multicolumn{2}{c|}{average \# of attempts} \\ \cline{2-5}
	& total & who got a correct answer eventually & over all users & before a correct answer\\
	\hline
	(b) & 102 & 93 & 4.08 & 1.79 \\
	(d) & 93 & 93 & 3.12 & 1.57 \\
	(e) & 100 & 95 & 5.24 & 3.45 \\
	(g) & 99 & 91 & 5.90 & 3.76 \\
	(i) & 120 & 94 & 11.10 & 7.46\\
	\hline
	\end{tabular}
	\caption{Statistics on \OurSys{} usage.}
	\label{table:user-study-stats}
\end{minipage}
	\vfil
\begin{minipage}{1.0\linewidth}
	\begin{tabular}{|>{\centering\arraybackslash}p{8em}|>{\centering\arraybackslash}p{4em}|>{\centering\arraybackslash}p{4em}|>{\centering\arraybackslash}p{4em}|>{\centering\arraybackslash}p{4em}|>{\centering\arraybackslash}p{4em}|>{\centering\arraybackslash}p{4em}|}
	\hline	
          Did the student use
          & \multirow{2}{*}{No} & \multirow{2}{*}{Yes} 
          & \multicolumn{4}{c|}{Time of the first use (before due date)}  \\ \cline{4-7}
          \OurSys{} for (i)? & & & 5-7 days  & 3-4 days & 2 days & 1 day \\ \hline
        \# of students & 49 & 120 &45 & 30 & 16 & 29 \\
        \hline
        Mean score on (i) & 89.80 & 94.40 & 97.14 & 99.05 & 91.96 & 86.70 \\
	Std.\ dev.\       & 30.58 & 19.00 & 15.41 & 5.22  & 25.54 & 26.16 \\
        \hline
        Mean score on (h) & 88.34 & 93.57 & 96.83 & 95.24 & 95.54 & 85.71 \\
	Std.\ dev.\       & 31.77 & 20.86 &  14.89 & 18.12 & 17.86 & 30.06 \\
	Mean score on (j) & 85.46 & 85.42 & 96.67 & 90.00 & 82.81 & 64.66 \\
	Std.\ dev.\       & 34.17 & 34.39 & 16.51 & 30.51 & 37.33 & 47.02 \\
	\hline
	\end{tabular}
	\caption{Comparison of performance on (h) and (j) between students whether they used \OurSys{} for (i) or not.}
	\label{table:user-study-stats-grade-i}
\vspace{-2mm}
\end{minipage}
\end{minipage}
\hfil
\begin{minipage}{0.38\linewidth}\centering
\begin{subfigure}
\centering
\includegraphics[width=\linewidth]{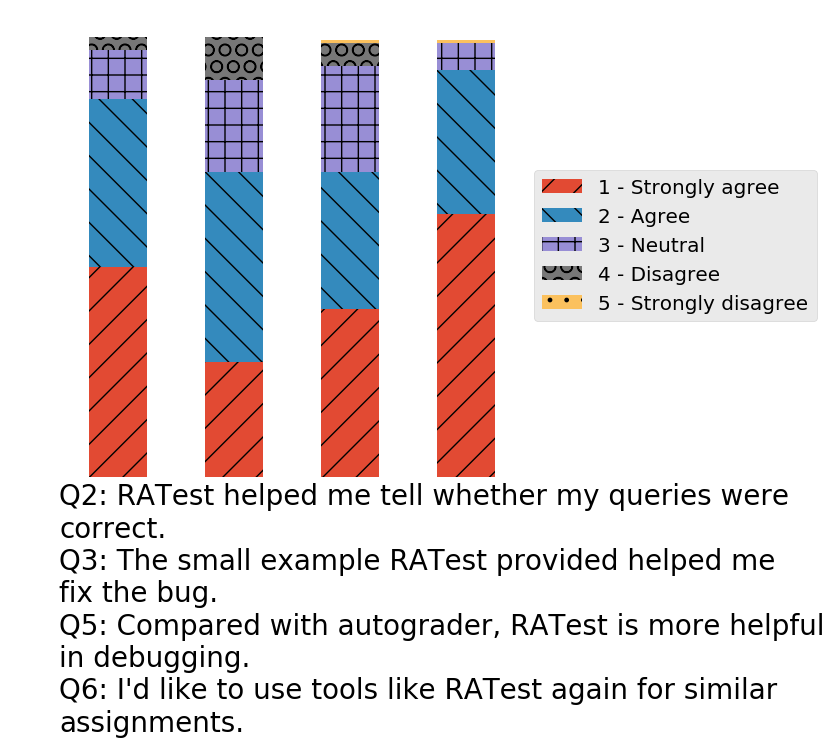}
\caption{Results of student feedback.}
\label{fig:user-study-feedback}
\end{subfigure} 
\end{minipage}
\end{minipage}
\vspace{-3mm}
\end{figure*}

\cut{
\begin{table*}[htb]\scriptsize{
	\centering
	\begin{tabular}{|>{\centering\arraybackslash}p{4em}|>{\centering\arraybackslash}p{3em}|>{\centering\arraybackslash}p{10em}|>{\centering\arraybackslash}p{4em}|>{\centering\arraybackslash}p{10em}|}
	\hline	
	\multirow{2}{*}{Problem} & \multicolumn{2}{c|}{\# of students} & \multicolumn{2}{c|}{Mean \# of attempts} \\ \cline{2-5}
	& in total & who made a ``correct'' submission on \OurSys{} & of all students & before the first ``correct'' submission on \OurSys{}\\
	\hline
	(b) & 102 & 93 & 4.08 & 1.79 \\
	(d) & 93 & 93 & 3.12 & 1.57 \\
	(e) & 100 & 95 & 5.24 & 3.45 \\
	(g) & 99 & 91 & 5.90 & 3.76 \\
	(i) & 120 & 94 & 11.10 & 7.46\\
	\hline
	\end{tabular}
	\caption{Cumulative Statistics of \OurSys{} Usage}
	\label{table:user-study-stats}}
\end{table*}}

\begin{table}[htb]
\scriptsize
	\centering
	\begin{tabular}{|>{\centering\arraybackslash}p{5em}|>{\centering\arraybackslash}p{8em}|>{\centering\arraybackslash}p{3em}|>{\centering\arraybackslash}p{3em}|}
	\hline	
	\multicolumn{2}{|p{13em}|}{Did the student use \OurSys{}?} &  No & Yes \\ \hline
	\multirow{2}{*}{\shortstack{Problem (b)}} & \# of students & 67 & 102 \\
	& Mean score & 100.00 & 100.00 \\
	& Std.\ dev.\ & 0.00  & 0.00 \\
	\hline
	\multirow{2}{*}{\shortstack{Problem (d)}} & \# of students & 76 & 93 \\
	& Mean score & 100.00 & 100.00 \\
	& Std.\ dev.\ & 0.00  & 0.00 \\
	\hline
	\multirow{2}{*}{\shortstack{Problem (e)}} & \# of students & 69 & 100 \\
	& Mean score & 99.03 & 99.67 \\
	& Std.\ dev.\ & 5.63 & 3.33 \\
	\hline
	\multirow{2}{*}{\shortstack{Problem (g)}} & \# of students & 70 & 99 \\
	& Mean score & 92.38 & 97.98 \\
	& Std.\ dev.\ & 26.11 & 14.14 \\
	\hline
	\multirow{2}{*}{\shortstack{Problem (i)}} & \# of students & 49 & 120 \\
	& Mean score & 89.80 & 94.40 \\
	& Std.\ dev.\ & 30.58 & 19.00 \\
	\hline
	\end{tabular}
	\caption{Comparison of performance between students who did
          not use \OurSys{} and those who did, on problems for which
          \OurSys{} was available.}
	\label{table:user-study-stats-grade}
 \vspace{-2mm}
\end{table}

\vspace*{-1ex}
\textbf{Quantitative Analysis of Usage Patterns and Scores.~}
Before exploring the impact of \OurSys{} on student scores, let us
examine some basic usage statistics, summarized in Figure~\ref{table:user-study-stats}.  Overall, $137$ students (more
than $80\%$ of the class) attempted a total of $3{,}146$ submissions
to \OurSys{}.  The sheer volume of the usage speaks to the demand for
tools like \OurSys{}, and the sustained usage (across problems)
suggests that the students found \OurSys{} useful.  It is also worth
noting that number of attempts reflects problem difficulty; for
example, (i), the most difficult problem, took far more attempts than
other problems.  We also note that while \OurSys{} helped the vast of
majority of its users get the correct queries in the end; some users
never did.  We observed in the usage log some unintended uses of
\OurSys{}: e.g., one student made more than a hundred incorrect
attempts on a problem, most of which contained basic errors (such as
syntax); apparently, \OurSys{} was used to just try queries out as
opposed to debugging queries after they failed the auto-grader.  Such
outliers explain the phenomenon shown in Figure~\ref{table:user-study-stats} where the overall average \# of
attempts were much higher than the average \# before a correct
attempt.

Next, we examine how the use of \OurSys{} helps improve student
scores.  Table~\ref{table:user-study-stats-grade} compares the scores
achieved by students who did not use \OurSys{} versus those who did,
on problem for which we made \OurSys{} available.  For simple problems
such as (b), (d), and (e), there is no little or no difference at all,
because nearly everybody got perfect scores with or without help from
\OurSys{}.  However, for more difficult problems, (g) and (i),
students who used \OurSys{} had a clear advantage, with average scores
improved from 92.38 to 97.98 and from 89.80 to 94.40, respectively.
Of course, within the constraints of our user study, it is still
difficult to conclude how much of this improvement comes from
\OurSys{} itself; it is conceivable that students who opted to use
\OurSys{} were simply more diligent and therefore would generally
perform better than others.  While we cannot definitively attribute
all improvement in student performance to \OurSys{}, we next provide
some evidence that it did help in a significant way.

\cut{
\begin{table*}[htb]\scriptsize
	\centering
	\begin{tabular}{|>{\centering\arraybackslash}p{8em}|>{\centering\arraybackslash}p{4em}|>{\centering\arraybackslash}p{4em}|>{\centering\arraybackslash}p{4em}|>{\centering\arraybackslash}p{4em}|>{\centering\arraybackslash}p{4em}|>{\centering\arraybackslash}p{4em}|}
	\hline	
          Did the student use
          & \multirow{2}{*}{No} & \multirow{2}{*}{Yes} 
          & \multicolumn{4}{c|}{Time of the first use (before due date)}  \\ \cline{4-7}
          \OurSys{} for (i)? & & & 5-7 days  & 3-4 days & 2 days & 1 day \\ \hline
        \# of students & 49 & 120 &45 & 30 & 16 & 29 \\
        \hline
        Mean score on (i) & 89.80 & 94.40 & 97.14 & 99.05 & 91.96 & 86.70 \\
	Std.\ dev.\       & 30.58 & 19.00 & 15.41 & 5.22  & 25.54 & 26.16 \\
        \hline
        Mean score on (h) & 88.34 & 93.57 & 96.83 & 95.24 & 95.54 & 85.71 \\
	Std.\ dev.\       & 31.77 & 20.86 &  14.89 & 18.12 & 17.86 & 30.06 \\
	Mean score on (j) & 85.46 & 85.42 & 96.67 & 90.00 & 82.81 & 64.66 \\
	Std.\ dev.\       & 34.17 & 34.39 & 16.51 & 30.51 & 37.33 & 47.02 \\
	\hline
	\end{tabular}
	\caption{Comparison of performance on (h) and (j) between students whether they used \OurSys{} for (i) or not.}
	\label{table:user-study-stats-grade-i}
\vspace{-2mm}
\end{table*}
}
\cut{
\begin{figure}
\begin{minipage}{1.0\linewidth}
  \begin{minipage}{0.49\linewidth}
	\centering
	\includegraphics[width=\linewidth]{../user_study/grade_g.png}
  \end{minipage}
  \hfill
  \begin{minipage}{0.49\linewidth}
	\centering
	\includegraphics[width=\linewidth]{../user_study/grade_i.png}
  \end{minipage}
  \caption{Students Grade Distribution by \OurSys{} Usage, bar indicates mean grade, and numbers in the x-axis indicates the days before the due when student started to use \OurSys{} on the problem}
  \label{fig:grade-ratest-usage}
  \end{minipage}
 \end{figure}
}

\cut{
As we expect, the statistical evaluation (see Table~\ref{table:user-study-stats-grade}) shows that here is a significant correlation between students' grades of this assignment and their usage of \OurSys{} \cut{(Figure~\ref{fig:grade-ratest-usage})}. We focus on problems (g) and (i) because they are the most challenging ones (almost all students get full score in the (b), (d), and (e)). We divide students into groups according to how early did they start to use \OurSys{} for the problem or they did not use \OurSys{} for the problem at all. 
Students who started to use \OurSys{} for question (g) at least one day earlier than the due (5-7 days:$n = 36, \mu = 100.00, \sigma = 0.00$; 3-4 days: $n=25, \mu = 96.00, \sigma = 20.00$; 2 days:$n = 16, \mu = 100.00, \sigma = 0.00$ )received higher grades compared to those who did not ($n = 70, \mu = 92.38, \sigma = 26.11$) --- two groups are of significant difference (5-7 days: $\nu=69, t=2.442, p < 0.025$; 2 days: $\nu=69, t=2.442, p < 0.025$). Moreover, students also received higher grade on (f) (5-7 days:$n = 36, \mu = 100.00, \sigma = 0.00$; 3-4 days: $n=25, \mu = 100.00, \sigma = 0.00$; 2 days:$n = 16, \mu = 100.00, \sigma = 0.00$), which is the similar problem with (g) but not available in \OurSys{}, if they used \OurSys{} for (g) at least one day earlier than the due, compared to those without trying (g) ($n = 70, \mu = 91.43, \sigma = 26.43$) --- also a significant difference ($\nu=69, t=2.714, p < 0.025$ for all three groups). 
We conjecture that the counterexample from \OurSys{} can also benefit students' understanding for the similar problems. We can draw similar conclusions on problem (i): (3-4 days:$n = 30, \mu = 99.05, \sigma = 5.22$) received higher grades compared to those who did not ($n = 49, \mu = 89.80, \sigma = 30.58$), which indicates a significant difference ($\nu=52, t=2.06, p < 0.05$). We notice that the improvement in grades are less significant compared with problem(g), and students even received lower grades on (i) and (h) if they started to try on \OurSys{} in the last day. It may be because (i) and (h) are harder than (g) and (f).}
\cut{
Below we describe one typical use case from the submission log of students, which is representative of the students' overall experience. A student submits an RA query for problem (i), the hardest one whose standard solution involves two difference operator, one join operator, and three projection operators over four input relations. The first time, s/he forgets the foreign key constraint and only involves three relations in the query, and then \OurSys{} returns a counterexample with only one tuple in the relation that the student forgets. Twenty seconds later, the student notices this issue and then add the relation to the query. However, the query only involves one difference operator and two joins, and \OurSys{} returns a counterexample with eight tuples showing that there is one result tuple in the standard output but not in the student's query output. For the third time, the student submits a query with two differences and two joins. Eight minutes later, the student finally submits a correct query. The use case provided evidence that \OurSys{} provides counterexamples as effective explanations, and quickly helps students debug their queries.
}

Here, we zoom in on the three most difficult problems, (h), (i), and
(j); \OurSys{} was only made available for (i).  Problem (h)
(\emph{find all drinkers who frequent only those bars that
  serve \textbf{some} beers they like}) is quite similar to (i) (the
difference being ``some beers'' vs.\ ``only beers'').  Problem (j)
(\emph{find all (bar1, bar2) pairs where the set of beer served at
  bar1 is a proper subset of those served at bar2}) on other hand
requires very different solution strategy.  Between those who did not
use \OurSys{} for (i) and those who did,
Figure~\ref{table:user-study-stats-grade-i} (focus on the first three
columns and ignore the rest for now) compares their scores on  (h)
and (j).  We see that for the similar problem (h), those who used
\OurSys{} on (i) significantly improved their scores on (h), with a
degree comparable to the improvement on (i).  For
the dissimilar problem (j), those who used \OurSys{} no (i) showed no
improvement in their scores on (j)---the two score distributions are
practically the same.  We make two observations here.
First, it is clear that not all improvement in student performance can
be explained by student ``diligence'' alone; otherwise we would have
seen higher performance on (j) for students who used \OurSys{} on (i).
Second, there is clearly a learning effect: using \OurSys{} for one
problem can help with a similar problem: (i) helps (h).

Figure~\ref{table:user-study-stats-grade-i}, in its last four columns,
also breaks down the statistics by when a student started to work on
Problem (i).  Not surprisingly, we see that ``procrastinators'' (those
who started very close to the due date) performed clearly worse than
others.  If somebody started to work on (i) using \OurSys{} only the
day before the homework was due, this individual would be expected to
perform even worse than an ``average'' student who opted not to use
\OurSys{} at all, especially for the last problem.  It would have been
nice if we can similarly break down the statistics for students who
opted not to use \OurSys{} at all, but it was not possible in that
case to know when they started to work on the problems.  We could only
conjecture that a similar trend might exist for procrastinators, so
using \OurSys{} did not hurt any individual's performance.

\textbf{Results of Anonymous Questionnaire.~}
\cut{
\begin{figure}[H]

\centering
\includegraphics[width=\linewidth]{../user_study/student_feedback.png}
\caption{Results of the student feedback.}
\label{fig:user-study-feedback}
\end{figure} 
}
We collected 134 valid responses to our anonymous questionnaire;
Figure~\ref{fig:user-study-feedback} summarizes these responses.  The
feedback was largely positive.  For instance, 69.4\% of the
respondents agree or strongly agree that the explanation by
counterexamples helped them understand or fix the bug in their
queries, and 93.2\% would like to use similar tools in the future for
assignments on querying databases.  We also asked students which
problems they found \OurSys{} to be most helpful (multiple choices
were allowed): 58\% voted for (g) and 94\% voted for (i),
which were indeed the most challenging ones.  We also solicited
open-ended comments on \OurSys{}.  These comments were overwhelming
positive and reinforces our conclusions from the quantitative
analysis, e.g.:
 \cut{``It was incredibly useful debugging edge cases in the larger
  dataset not provided in our sample dataset with behavior not
  explicitly described in the problem set.''
 ``Overall, very helpful and would like to see similar testers
  for future assignments.'' ``I liked how it gave us a concise example showing what we did
  wrong.''}

\begin{itemize}\itshape
\vspace{-2mm}
\item ``It was incredibly useful debugging edge cases in the larger
  dataset not provided in our sample dataset with behavior not
  explicitly described in the problem set.''
\item ``Overall, very helpful and would like to see similar testers
  for future assignments.''
\item ``I liked how it gave us a concise example showing what we did
  wrong.''
 \vspace{-3mm}
\end{itemize}

\textbf{Summary.~}
Overall, the conclusion of our user study is positive.  Students who
used \OurSys{} did better, and their improvement cannot be attributed
all to merely the fact that they opted to use an additional
tool---\OurSys{} did add real value.  Also, using \OurSys{} on
one problem could also help with another problem,
provided that the problems are similar.  Finally, most students
found \OurSys{} very useful and would like to use similar systems in
the future.


\vspace{-2mm}
\section{Related Work}
\label{sec:related-work}

\textbf{Test data generation.~} Cosette\cite{chu2017cosette}, which targets at deciding SQL equivalence without any test instances, encodes SQL queries to constraints using symbolic execution, and uses a constraint solver to find counterexamples over which the two queries return different results. Cosette uses incremental solving to dynamically increase the size of each symbolic relation, thus it will return counterexamples with least number of distinct tuples, but the total number of tuples is not minimized. ALso, it deals with only integer domain and returns counterexamples of arbitary values, which may be hard for people to read. 
XData\cite{chandra2015data} generates test data by covering different types of query mutants of the standard query, without looking into wrong queries. Qex\cite{veanes2010qex} is a tool for generating input relations and parameter values for a parameterized SQL query that also uses the SMT solver Z3, which aims at unit testing of SQL queries. It does not support nested queries and set operations and hence it cannot work for our problem because of our use of difference. 
\par
\textbf{Provenance and witness.~} 
Data provenance has primarily been studied for non-aggregate queries: Buneman et al.\cite{buneman2001and} defined why-provenance of an output tuple in the result set, which they call the witness basis. Green et al.\cite{green2007provenance} introduced how-provenance with the general framework of \emph{provenance semiring}. Sarma et al.\cite{sarma2008exploiting} gave algorithms for computing how-provenance over various RA operators in the Trio system. Amsterdamer et al. \cite{amsterdamer2011provenance} extended the \emph{provenance semiring} framework\cite{green2007provenance} to support aggregate queries. Besides these theoretical works, there are systems that capture different forms of provenance \cite{karvounarakis2010querying, glavic2009perm, arab2014generic, green2007update, psallidas2018smoke, senellart2018provsql}. However, to the best of our knowledge,
no prior work considered \SWP/\SCP, and there are no systems available that support provenance for general SPJUD and aggregate queries.
\par
\textbf{Teaching or grading tool for programming.~}
Due to popularity of students taking programming-related courses, teaching and grading tools for programming assignments that automatically generate feedback for submissions are receiving a lot of attention\cite{parihar2017automatic, kaleeswaran2016semi, gupta2017deepfix}. In the database community, Chandra et al. built XData\cite{chandra2015data} that can be used for grading by generating multiple test cases for different query mutants, as well as giving immediate feedback to student. The latter is similar to our \OurSys{} tool. Jiang and Nandi\cite{jiang2015designing, nandi2015breathing} designed and prototyped interactive electronic textbook to help students get rapid feedbacks from querying the database with novel interaction techniques.
\par
\textbf{Explanations for query answers.~}
Explanations based on tuples in the provenance has been recently studied by Wu-Madden~\cite{WM13} and Roy-Suciu~\cite{RS14}. 
These works take an aggregate query and a user question as input, find tuples whose removal will change the answer in the opposite direction, and returns a compact summary as explanations. 





\balance


\bibliographystyle{ACM-Reference-Format}
\bibliography{grading_main}  


\begin{thebibliography}{39}


\ifx \showCODEN    \undefined \def \showCODEN     #1{\unskip}     \fi
\ifx \showDOI      \undefined \def \showDOI       #1{#1}\fi
\ifx \showISBNx    \undefined \def \showISBNx     #1{\unskip}     \fi
\ifx \showISBNxiii \undefined \def \showISBNxiii  #1{\unskip}     \fi
\ifx \showISSN     \undefined \def \showISSN      #1{\unskip}     \fi
\ifx \showLCCN     \undefined \def \showLCCN      #1{\unskip}     \fi
\ifx \shownote     \undefined \def \shownote      #1{#1}          \fi
\ifx \showarticletitle \undefined \def \showarticletitle #1{#1}   \fi
\ifx \showURL      \undefined \def \showURL       {\relax}        \fi
\providecommand\bibfield[2]{#2}
\providecommand\bibinfo[2]{#2}
\providecommand\natexlab[1]{#1}
\providecommand\showeprint[2][]{arXiv:#2}

\bibitem[\protect\citeauthoryear{Amsterdamer, Deutch, and Tannen}{Amsterdamer
  et~al\mbox{.}}{2011a}]%
        {amsterdamer2011limitations}
\bibfield{author}{\bibinfo{person}{Yael Amsterdamer}, \bibinfo{person}{Daniel
  Deutch}, {and} \bibinfo{person}{Val Tannen}.}
  \bibinfo{year}{2011}\natexlab{a}.
\newblock \showarticletitle{On the Limitations of Provenance for Queries with
  Difference}. In \bibinfo{booktitle}{{\em TaPP}}.
\newblock


\bibitem[\protect\citeauthoryear{Amsterdamer, Deutch, and Tannen}{Amsterdamer
  et~al\mbox{.}}{2011b}]%
        {amsterdamer2011provenance}
\bibfield{author}{\bibinfo{person}{Yael Amsterdamer}, \bibinfo{person}{Daniel
  Deutch}, {and} \bibinfo{person}{Val Tannen}.}
  \bibinfo{year}{2011}\natexlab{b}.
\newblock \showarticletitle{Provenance for aggregate queries}. In
  \bibinfo{booktitle}{{\em PODS}}. \bibinfo{pages}{153--164}.
\newblock


\bibitem[\protect\citeauthoryear{Anonymous}{Anonymous}{2017}]%
        {anonymous2017rai}
\bibfield{author}{\bibinfo{person}{Anonymous}.}
  \bibinfo{year}{2017}\natexlab{}.
\newblock \bibinfo{title}{A Relational Algebra Interpreter}.
\newblock   (\bibinfo{year}{2017}).
\newblock


\bibitem[\protect\citeauthoryear{Arab, Gawlick, Radhakrishnan, Guo, and
  Glavic}{Arab et~al\mbox{.}}{2014}]%
        {arab2014generic}
\bibfield{author}{\bibinfo{person}{Bahareh Arab}, \bibinfo{person}{Dieter
  Gawlick}, \bibinfo{person}{Venkatesh Radhakrishnan}, \bibinfo{person}{Hao
  Guo}, {and} \bibinfo{person}{Boris Glavic}.} \bibinfo{year}{2014}\natexlab{}.
\newblock \showarticletitle{A generic provenance middleware for database
  queries, updates, and transactions}. In \bibinfo{booktitle}{{\em {TaPP}}}.
\newblock


\bibitem[\protect\citeauthoryear{Barrett, Conway, Deters,
  et~al\mbox{.}}{Barrett et~al\mbox{.}}{2011}]%
        {barrett2011cvc4}
\bibfield{author}{\bibinfo{person}{Clark Barrett},
  \bibinfo{person}{Christopher~L. Conway}, \bibinfo{person}{Morgan Deters},
  {et~al\mbox{.}}} \bibinfo{year}{2011}\natexlab{}.
\newblock \showarticletitle{{CVC4}}. In \bibinfo{booktitle}{{\em CAV '11}},
  Vol.~\bibinfo{volume}{6806}. \bibinfo{publisher}{Springer},
  \bibinfo{pages}{171--177}.
\newblock


\bibitem[\protect\citeauthoryear{Barrett, Stump, Tinelli,
  et~al\mbox{.}}{Barrett et~al\mbox{.}}{2010}]%
        {barrett2010smt}
\bibfield{author}{\bibinfo{person}{Clark Barrett}, \bibinfo{person}{Aaron
  Stump}, \bibinfo{person}{Cesare Tinelli}, {et~al\mbox{.}}}
  \bibinfo{year}{2010}\natexlab{}.
\newblock \showarticletitle{The smt-lib standard: Version 2.0}. In
  \bibinfo{booktitle}{{\em Proceedings of the 8th International Workshop on
  Satisfiability Modulo Theories}}, Vol.~\bibinfo{volume}{13}.
  \bibinfo{pages}{14}.
\newblock


\bibitem[\protect\citeauthoryear{Barrett and Tinelli}{Barrett and
  Tinelli}{2018}]%
        {barrett2018satisfiability}
\bibfield{author}{\bibinfo{person}{Clark Barrett} {and} \bibinfo{person}{Cesare
  Tinelli}.} \bibinfo{year}{2018}\natexlab{}.
\newblock \showarticletitle{Satisfiability modulo theories}.
\newblock In \bibinfo{booktitle}{{\em Handbook of Model Checking}}.
  \bibinfo{publisher}{Springer}, \bibinfo{pages}{305--343}.
\newblock


\bibitem[\protect\citeauthoryear{Biere}{Biere}{[n. d.]}]%
        {CaDiCaL2018}
\bibfield{author}{\bibinfo{person}{Armin Biere}.} \bibinfo{year}{[n.
  d.]}\natexlab{}.
\newblock \bibinfo{title}{{CaDiCaL: Simplified Satisfiability Solver}}.
\newblock \bibinfo{howpublished}{\url{https://github.com/arminbiere/cadical}}.
   (\bibinfo{year}{[n. d.]}).
\newblock
\newblock
\shownote{[Online; accessed 24-Oct-2018].}


\bibitem[\protect\citeauthoryear{Bj{\o}rner, Phan, and Fleckenstein}{Bj{\o}rner
  et~al\mbox{.}}{2015}]%
        {bjorner2015nuz}
\bibfield{author}{\bibinfo{person}{Nikolaj Bj{\o}rner},
  \bibinfo{person}{Anh-Dung Phan}, {and} \bibinfo{person}{Lars Fleckenstein}.}
  \bibinfo{year}{2015}\natexlab{}.
\newblock \showarticletitle{$\nu$Z-an optimizing SMT solver}. In
  \bibinfo{booktitle}{{\em International Conference on Tools and Algorithms for
  the Construction and Analysis of Systems}}. Springer,
  \bibinfo{pages}{194--199}.
\newblock


\bibitem[\protect\citeauthoryear{Buneman, Khanna, and Wang-Chiew}{Buneman
  et~al\mbox{.}}{2001}]%
        {buneman2001and}
\bibfield{author}{\bibinfo{person}{Peter Buneman}, \bibinfo{person}{Sanjeev
  Khanna}, {and} \bibinfo{person}{Tan Wang-Chiew}.}
  \bibinfo{year}{2001}\natexlab{}.
\newblock \showarticletitle{Why and where: A characterization of data
  provenance}. In \bibinfo{booktitle}{{\em International conference on database
  theory}}. Springer, \bibinfo{pages}{316--330}.
\newblock


\bibitem[\protect\citeauthoryear{Chandra, Chawda, Kar, Reddy, Shah, and
  Sudarshan}{Chandra et~al\mbox{.}}{2015}]%
        {chandra2015data}
\bibfield{author}{\bibinfo{person}{Bikash Chandra}, \bibinfo{person}{Bhupesh
  Chawda}, \bibinfo{person}{Biplab Kar}, \bibinfo{person}{KV~Maheshwara Reddy},
  \bibinfo{person}{Shetal Shah}, {and} \bibinfo{person}{S Sudarshan}.}
  \bibinfo{year}{2015}\natexlab{}.
\newblock \showarticletitle{Data generation for testing and grading SQL
  queries}.
\newblock \bibinfo{journal}{{\em The VLDB Journal\/}} \bibinfo{volume}{24},
  \bibinfo{number}{6} (\bibinfo{year}{2015}), \bibinfo{pages}{731--755}.
\newblock


\bibitem[\protect\citeauthoryear{Chu, Wang, Weitz, and Cheung}{Chu
  et~al\mbox{.}}{2017}]%
        {chu2017cosette}
\bibfield{author}{\bibinfo{person}{Shumo Chu}, \bibinfo{person}{Chenglong
  Wang}, \bibinfo{person}{Konstantin Weitz}, {and} \bibinfo{person}{Alvin
  Cheung}.} \bibinfo{year}{2017}\natexlab{}.
\newblock \showarticletitle{Cosette: An Automated Prover for SQL.}. In
  \bibinfo{booktitle}{{\em CIDR}}.
\newblock


\bibitem[\protect\citeauthoryear{Cohen, Sagiv, and Nutt}{Cohen
  et~al\mbox{.}}{2005}]%
        {cohen2005equivalences}
\bibfield{author}{\bibinfo{person}{Sara Cohen}, \bibinfo{person}{Yehoshua
  Sagiv}, {and} \bibinfo{person}{Werner Nutt}.}
  \bibinfo{year}{2005}\natexlab{}.
\newblock \showarticletitle{Equivalences among aggregate queries with
  negation}.
\newblock \bibinfo{journal}{{\em ACM Transactions on Computational Logic
  (TOCL)\/}} \bibinfo{volume}{6}, \bibinfo{number}{2} (\bibinfo{year}{2005}),
  \bibinfo{pages}{328--360}.
\newblock


\bibitem[\protect\citeauthoryear{Council}{Council}{2008}]%
        {council2008tpc}
\bibfield{author}{\bibinfo{person}{Transaction Processing~Performance
  Council}.} \bibinfo{year}{2008}\natexlab{}.
\newblock \showarticletitle{TPC-H benchmark specification}.
\newblock \bibinfo{journal}{{\em Published at http://www.tcp.org/hspec.html\/}}
   \bibinfo{volume}{21} (\bibinfo{year}{2008}), \bibinfo{pages}{592--603}.
\newblock


\bibitem[\protect\citeauthoryear{De~Moura and Bj{\o}rner}{De~Moura and
  Bj{\o}rner}{2008}]%
        {de2008z3}
\bibfield{author}{\bibinfo{person}{Leonardo De~Moura} {and}
  \bibinfo{person}{Nikolaj Bj{\o}rner}.} \bibinfo{year}{2008}\natexlab{}.
\newblock \showarticletitle{Z3: An efficient SMT solver}. In
  \bibinfo{booktitle}{{\em International conference on Tools and Algorithms for
  the Construction and Analysis of Systems}}. \bibinfo{pages}{337--340}.
\newblock


\bibitem[\protect\citeauthoryear{Garey, Johnson, and Stockmeyer}{Garey
  et~al\mbox{.}}{1976}]%
        {garey1976some}
\bibfield{author}{\bibinfo{person}{Michael~R Garey}, \bibinfo{person}{David~S.
  Johnson}, {and} \bibinfo{person}{Larry Stockmeyer}.}
  \bibinfo{year}{1976}\natexlab{}.
\newblock \showarticletitle{Some simplified NP-complete graph problems}.
\newblock \bibinfo{journal}{{\em Theoretical computer science\/}}
  \bibinfo{volume}{1}, \bibinfo{number}{3} (\bibinfo{year}{1976}),
  \bibinfo{pages}{237--267}.
\newblock


\bibitem[\protect\citeauthoryear{Glavic and Alonso}{Glavic and Alonso}{2009}]%
        {glavic2009perm}
\bibfield{author}{\bibinfo{person}{Boris Glavic} {and} \bibinfo{person}{Gustavo
  Alonso}.} \bibinfo{year}{2009}\natexlab{}.
\newblock \showarticletitle{Perm: Processing provenance and data on the same
  data model through query rewriting}. In \bibinfo{booktitle}{{\em ICDE}}.
  \bibinfo{pages}{174--185}.
\newblock


\bibitem[\protect\citeauthoryear{Green, Karvounarakis, Ives, and Tannen}{Green
  et~al\mbox{.}}{2007b}]%
        {green2007update}
\bibfield{author}{\bibinfo{person}{Todd~J Green}, \bibinfo{person}{Grigoris
  Karvounarakis}, \bibinfo{person}{Zachary~G Ives}, {and} \bibinfo{person}{Val
  Tannen}.} \bibinfo{year}{2007}\natexlab{b}.
\newblock \showarticletitle{Update exchange with mappings and provenance}. In
  \bibinfo{booktitle}{{\em PVLDB}}. \bibinfo{pages}{675--686}.
\newblock


\bibitem[\protect\citeauthoryear{Green, Karvounarakis, and Tannen}{Green
  et~al\mbox{.}}{2007a}]%
        {green2007provenance}
\bibfield{author}{\bibinfo{person}{Todd~J Green}, \bibinfo{person}{Grigoris
  Karvounarakis}, {and} \bibinfo{person}{Val Tannen}.}
  \bibinfo{year}{2007}\natexlab{a}.
\newblock \showarticletitle{Provenance semirings}. In \bibinfo{booktitle}{{\em
  PODS}}. \bibinfo{pages}{31--40}.
\newblock


\bibitem[\protect\citeauthoryear{Gupta, Pal, Kanade, and Shevade}{Gupta
  et~al\mbox{.}}{2017}]%
        {gupta2017deepfix}
\bibfield{author}{\bibinfo{person}{Rahul Gupta}, \bibinfo{person}{Soham Pal},
  \bibinfo{person}{Aditya Kanade}, {and} \bibinfo{person}{Shirish Shevade}.}
  \bibinfo{year}{2017}\natexlab{}.
\newblock \showarticletitle{DeepFix: Fixing Common C Language Errors by Deep
  Learning.}. In \bibinfo{booktitle}{{\em AAAI}}. \bibinfo{pages}{1345--1351}.
\newblock


\bibitem[\protect\citeauthoryear{Imieli\'{n}ski and Lipski}{Imieli\'{n}ski and
  Lipski}{[n. d.]}]%
        {Imielinski:1984:IIR:1634.1886}
\bibfield{author}{\bibinfo{person}{Tomasz Imieli\'{n}ski} {and}
  \bibinfo{person}{Witold Lipski, Jr.}} \bibinfo{year}{[n. d.]}\natexlab{}.
\newblock \showarticletitle{Incomplete Information in Relational Databases}.
\newblock \bibinfo{journal}{{\em J. ACM\/}} \bibinfo{volume}{31},
  \bibinfo{number}{4} (\bibinfo{year}{[n. d.]}), \bibinfo{pages}{761--791}.
\newblock
\showISSN{0004-5411}


\bibitem[\protect\citeauthoryear{Jiang and Nandi}{Jiang and Nandi}{2015}]%
        {jiang2015designing}
\bibfield{author}{\bibinfo{person}{Lilong Jiang} {and} \bibinfo{person}{Arnab
  Nandi}.} \bibinfo{year}{2015}\natexlab{}.
\newblock \showarticletitle{Designing interactive query interfaces to teach
  database systems in the classroom}. In \bibinfo{booktitle}{{\em Proceedings
  of the 33rd Annual ACM Conference Extended Abstracts on Human Factors in
  Computing Systems}}. \bibinfo{pages}{1479--1482}.
\newblock


\bibitem[\protect\citeauthoryear{Kaleeswaran, Santhiar, Kanade, and
  Gulwani}{Kaleeswaran et~al\mbox{.}}{2016}]%
        {kaleeswaran2016semi}
\bibfield{author}{\bibinfo{person}{Shalini Kaleeswaran},
  \bibinfo{person}{Anirudh Santhiar}, \bibinfo{person}{Aditya Kanade}, {and}
  \bibinfo{person}{Sumit Gulwani}.} \bibinfo{year}{2016}\natexlab{}.
\newblock \showarticletitle{Semi-supervised verified feedback generation}. In
  \bibinfo{booktitle}{{\em SIGSOFT}}. \bibinfo{pages}{739--750}.
\newblock


\bibitem[\protect\citeauthoryear{Karvounarakis, Ives, and Tannen}{Karvounarakis
  et~al\mbox{.}}{2010}]%
        {karvounarakis2010querying}
\bibfield{author}{\bibinfo{person}{Grigoris Karvounarakis},
  \bibinfo{person}{Zachary~G Ives}, {and} \bibinfo{person}{Val Tannen}.}
  \bibinfo{year}{2010}\natexlab{}.
\newblock \showarticletitle{Querying data provenance}. In
  \bibinfo{booktitle}{{\em SIGMOD}}. \bibinfo{pages}{951--962}.
\newblock


\bibitem[\protect\citeauthoryear{Kratsch, Marx, and Wahlstr{\"o}m}{Kratsch
  et~al\mbox{.}}{2010}]%
        {kratsch2010parameterized}
\bibfield{author}{\bibinfo{person}{Stefan Kratsch}, \bibinfo{person}{D{\'a}niel
  Marx}, {and} \bibinfo{person}{Magnus Wahlstr{\"o}m}.}
  \bibinfo{year}{2010}\natexlab{}.
\newblock \showarticletitle{Parameterized complexity and kernelizability of max
  ones and exact ones problems}. In \bibinfo{booktitle}{{\em MFCS}}.
  \bibinfo{pages}{489--500}.
\newblock


\bibitem[\protect\citeauthoryear{Ley and Dagstuhl}{Ley and Dagstuhl}{2018}]%
        {dblpdata}
\bibfield{author}{\bibinfo{person}{Michael Ley} {and} \bibinfo{person}{Schloss
  Dagstuhl}.} \bibinfo{year}{2018}\natexlab{}.
\newblock \showarticletitle{DBLP database}.
\newblock \bibinfo{howpublished}{\url{https://dblp.uni-trier.de/xml/}}.
\newblock  (\bibinfo{year}{2018}).
\newblock


\bibitem[\protect\citeauthoryear{Li, Albarghouthi, Kincaid, Gurfinkel, and
  Chechik}{Li et~al\mbox{.}}{2014}]%
        {li2014symbolic}
\bibfield{author}{\bibinfo{person}{Yi Li}, \bibinfo{person}{Aws Albarghouthi},
  \bibinfo{person}{Zachary Kincaid}, \bibinfo{person}{Arie Gurfinkel}, {and}
  \bibinfo{person}{Marsha Chechik}.} \bibinfo{year}{2014}\natexlab{}.
\newblock \showarticletitle{Symbolic optimization with SMT solvers}. In
  \bibinfo{booktitle}{{\em ACM SIGPLAN Notices}}, Vol.~\bibinfo{volume}{49}.
  ACM, \bibinfo{pages}{607--618}.
\newblock


\bibitem[\protect\citeauthoryear{Nandi}{Nandi}{2015}]%
        {nandi2015breathing}
\bibfield{author}{\bibinfo{person}{Arnab Nandi}.}
  \bibinfo{year}{2015}\natexlab{}.
\newblock \showarticletitle{Breathing Life into Database Textbooks.}. In
  \bibinfo{booktitle}{{\em CIDR}}.
\newblock


\bibitem[\protect\citeauthoryear{Nutt, Sagiv, and Shurin}{Nutt
  et~al\mbox{.}}{1998}]%
        {nutt1998deciding}
\bibfield{author}{\bibinfo{person}{Werner Nutt}, \bibinfo{person}{Yehoshus
  Sagiv}, {and} \bibinfo{person}{Sara Shurin}.}
  \bibinfo{year}{1998}\natexlab{}.
\newblock \showarticletitle{Deciding equivalences among aggregate queries}. In
  \bibinfo{booktitle}{{\em PODS}}. \bibinfo{pages}{214--223}.
\newblock


\bibitem[\protect\citeauthoryear{Parihar, Dadachanji, Singh, Das, Karkare, and
  Bhattacharya}{Parihar et~al\mbox{.}}{2017}]%
        {parihar2017automatic}
\bibfield{author}{\bibinfo{person}{Sagar Parihar}, \bibinfo{person}{Ziyaan
  Dadachanji}, \bibinfo{person}{Praveen~Kumar Singh}, \bibinfo{person}{Rajdeep
  Das}, \bibinfo{person}{Amey Karkare}, {and} \bibinfo{person}{Arnab
  Bhattacharya}.} \bibinfo{year}{2017}\natexlab{}.
\newblock \showarticletitle{Automatic grading and feedback using program repair
  for introductory programming courses}. In \bibinfo{booktitle}{{\em
  Proceedings of the 2017 ACM Conference on Innovation and Technology in
  Computer Science Education}}. ACM, \bibinfo{pages}{92--97}.
\newblock


\bibitem[\protect\citeauthoryear{Psallidas and Wu}{Psallidas and Wu}{2018}]%
        {psallidas2018smoke}
\bibfield{author}{\bibinfo{person}{Fotis Psallidas} {and}
  \bibinfo{person}{Eugene Wu}.} \bibinfo{year}{2018}\natexlab{}.
\newblock \showarticletitle{Smoke: Fine-grained lineage at interactive speed}.
\newblock \bibinfo{journal}{{\em PVLDB\/}} \bibinfo{volume}{11},
  \bibinfo{number}{6} (\bibinfo{year}{2018}), \bibinfo{pages}{719--732}.
\newblock


\bibitem[\protect\citeauthoryear{Roy, Perduca, and Tannen}{Roy
  et~al\mbox{.}}{2011}]%
        {DBLP:conf/icdt/RoyPT11}
\bibfield{author}{\bibinfo{person}{Sudeepa Roy}, \bibinfo{person}{Vittorio
  Perduca}, {and} \bibinfo{person}{Val Tannen}.}
  \bibinfo{year}{2011}\natexlab{}.
\newblock \showarticletitle{Faster query answering in probabilistic databases
  using read-once functions}. In \bibinfo{booktitle}{{\em ICDT}}.
  \bibinfo{pages}{232--243}.
\newblock


\bibitem[\protect\citeauthoryear{Roy and Suciu}{Roy and Suciu}{2014}]%
        {RS14}
\bibfield{author}{\bibinfo{person}{Sudeepa Roy} {and} \bibinfo{person}{Dan
  Suciu}.} \bibinfo{year}{2014}\natexlab{}.
\newblock \showarticletitle{A formal approach to finding explanations for
  database queries}. In \bibinfo{booktitle}{{\em {SIGMOD}}}.
  \bibinfo{pages}{1579--1590}.
\newblock


\bibitem[\protect\citeauthoryear{Sarma, Theobald, and Widom}{Sarma
  et~al\mbox{.}}{2008}]%
        {sarma2008exploiting}
\bibfield{author}{\bibinfo{person}{Anish~Das Sarma}, \bibinfo{person}{Martin
  Theobald}, {and} \bibinfo{person}{Jennifer Widom}.}
  \bibinfo{year}{2008}\natexlab{}.
\newblock \showarticletitle{Exploiting lineage for confidence computation in
  uncertain and probabilistic databases}. In \bibinfo{booktitle}{{\em ICDE}}.
  IEEE, \bibinfo{pages}{1023--1032}.
\newblock


\bibitem[\protect\citeauthoryear{Senellart, Jachiet, Maniu, and
  Ramusat}{Senellart et~al\mbox{.}}{2018}]%
        {senellart2018provsql}
\bibfield{author}{\bibinfo{person}{Pierre Senellart}, \bibinfo{person}{Louis
  Jachiet}, \bibinfo{person}{Silviu Maniu}, {and} \bibinfo{person}{Yann
  Ramusat}.} \bibinfo{year}{2018}\natexlab{}.
\newblock \showarticletitle{ProvSQL: provenance and probability management in
  postgreSQL}.
\newblock \bibinfo{journal}{{\em PVLDB\/}} \bibinfo{volume}{11},
  \bibinfo{number}{12} (\bibinfo{year}{2018}), \bibinfo{pages}{2034--2037}.
\newblock


\bibitem[\protect\citeauthoryear{S{\"o}rensson and E{\'e}n}{S{\"o}rensson and
  E{\'e}n}{2009}]%
        {sorensson2009minisat}
\bibfield{author}{\bibinfo{person}{Niklas S{\"o}rensson} {and}
  \bibinfo{person}{Niklas E{\'e}n}.} \bibinfo{year}{2009}\natexlab{}.
\newblock \showarticletitle{Minisat 2.1 and minisat++ 1.0-sat race 2008
  editions}.
\newblock \bibinfo{journal}{{\em SAT\/}} (\bibinfo{year}{2009}),
  \bibinfo{pages}{31}.
\newblock


\bibitem[\protect\citeauthoryear{Vardi}{Vardi}{1982}]%
        {vardi1982complexity}
\bibfield{author}{\bibinfo{person}{Moshe~Y Vardi}.}
  \bibinfo{year}{1982}\natexlab{}.
\newblock \showarticletitle{The complexity of relational query languages}. In
  \bibinfo{booktitle}{{\em STOC}}. \bibinfo{pages}{137--146}.
\newblock


\bibitem[\protect\citeauthoryear{Veanes, Tillmann, and De~Halleux}{Veanes
  et~al\mbox{.}}{2010}]%
        {veanes2010qex}
\bibfield{author}{\bibinfo{person}{Margus Veanes}, \bibinfo{person}{Nikolai
  Tillmann}, {and} \bibinfo{person}{Jonathan De~Halleux}.}
  \bibinfo{year}{2010}\natexlab{}.
\newblock \showarticletitle{Qex: Symbolic SQL query explorer}. In
  \bibinfo{booktitle}{{\em International Conference on Logic for Programming
  Artificial Intelligence and Reasoning}}. Springer, \bibinfo{pages}{425--446}.
\newblock


\bibitem[\protect\citeauthoryear{Wu and Madden}{Wu and Madden}{2013}]%
        {WM13}
\bibfield{author}{\bibinfo{person}{Eugene Wu} {and} \bibinfo{person}{Samuel
  Madden}.} \bibinfo{year}{2013}\natexlab{}.
\newblock \showarticletitle{Scorpion: Explaining Away Outliers in Aggregate
  Queries}.
\newblock \bibinfo{journal}{{\em {PVLDB}\/}} \bibinfo{volume}{6},
  \bibinfo{number}{8} (\bibinfo{year}{2013}), \bibinfo{pages}{553--564}.
\newblock


\end{thebibliography}
\clearpage
\begin{appendix}

\section{Proofs of Theorems in Section~\ref{sec:spjud}}\label{sec:spjud-proof}

We will give the proofs of theorems in Table~\ref{table:spjud-complexity} in this section.

\subsection{SJ and SPU Queries}\label{sec:SJ-SPU}

Given $t \in Q_1(D) \setminus Q_2(D)$, the poly-time algorithm for SJ and SPU queries involve finding a smallest witness of $t$ in $D$ for $Q_1$, and using the fact that $Q_2$ is monotone and $t \notin Q_2(D)$, $\forall D' \subseteq D$, $t \notin Q_2(D')$.
\begin{Theorem}
\label{thm:sj-poly}
The SWP for two SJ queries is poly-time solvable in combined complexity.
\end{Theorem}

\begin{proof} Let $R_1, ..., R_k$ be all the relations that participate in the SJ query $Q_1$. 
For each relation $R_i$, $i \in [1, k]$, there must exist exactly one tuple $t_i = t.R_i$ (the $R_i$ component of $t$), which is part of the witness of $t$ (under set semantic). 
Since each $t_i$ must satisfy all selection conditions for $t$ to appear in $Q_1(D)$, the set $D_t = \{t_i | i \in [1, k]\}$ ensures that $t \in Q_1(D_t)$, and must be minimal.
Since $Q_2$ is monotone and $t \notin Q_2(D)$, we have $t \notin Q_2(D_t)$; hence $t \in (Q_2 - Q_1)(D_t)$. The running time to find $D_t$ is polynomial in $k$, giving polynomial combined complexity. 
\end{proof}


When projection is allowed, an output tuple may have multiple minimal witnesses, and we pick any one of them.

\begin{Theorem}
\label{thm:spu-poly}
The SWP for two SPU queries is polynomial-time solvable in combined complexity.
\end{Theorem}

\begin{proof}
We first consider SP queries. Given an output tuple $t$ in $Q_1(D)$, we scan the input relation to find a tuple $t'$ that satisfies the selection condition and whose projected attributes equal to $t$. The smallest witness $D_t$ only consists of only $t'$. 
For SPU queries, we do the same procedure as SP queries. At least one relation will return $t'$. Since $Q_2$ is monotone and $t \notin Q_2(D)$, we have $t \notin Q_2(D_t)$. The running time to find $D_t = \{t'\}$ is polynomial in $k$. 
\end{proof}



\vspace{-2mm}
\subsection{PJ Queries}\label{sec:PJ}
For queries involving both projection and join, we show that it is \NPhard{} in query complexity to find the smallest witness, even when the query can be evaluated in poly-time.
\begin{Theorem}
\label{thm:PJ-hard}
The SWP for two PJ queries is \NPhard{} in query complexity.
\end{Theorem}

\cut{
\begin{figure*}[!h]
\begin{minipage}{0.92\linewidth}\centering
\begin{minipage}{0.52\linewidth}\centering
\begin{minipage}{0.48\linewidth}\centering
\begin{subfigure}
    \centering
      \resizebox{\linewidth}{!}{
      \begin{tikzpicture}
        [scale=.8,auto=left,every node/.style={circle,fill=blue!20,scale=1.5}]
        \node (n6) at (10,10) {v6};
        \node (n4) at (3,10)  {v4};
        \node (n5) at (7,10)  {v5};
        \node (n1) at (1,8) {v1};
        \node (n2) at (3,6)  {v2};
        \node (n3) at (7,6)  {v3};

        \foreach \from/\to/\label in {n2/n1/e1,n2/n3/e2,n5/n3/e3,n4/n5/e4,n5/n6/e5,n1/n4/e6,n4/n2/e7}
          \draw (\from) -- (\to) node [midway, fill=white] {\large \label};

    \end{tikzpicture}
 }
    \end{subfigure}
\end{minipage}
\begin{minipage}{0.48\linewidth}\centering
    \begin{subfigure}
    \centering
        \resizebox{\linewidth}{!}{
\begin{tikzpicture}
  [scale=.8,auto=left,every node/.style={circle,fill=blue!20,scale=1.5}]
  \node (n6) at (10,10) {v6};
  \node (n4) at (3,10)  {v4};
  \node[fill=red] (n5) at (7,10)  {v5};
  \node[fill=red] (n1) at (1,8) {v1};
  \node[fill=red] (n2) at (3,6)  {v2};
  \node (n3) at (7,6)  {v3};

  \foreach \from/\to/\label in {n2/n1/e1,n2/n3/e2,n5/n3/e3,n4/n5/e4,n5/n6/e5,n1/n4/e6,n4/n2/e7}
    \draw (\from) -- (\to) node [midway, fill=white] {\large \label};

\end{tikzpicture}
        }
    \end{subfigure}
  \end{minipage}
 \caption{\small{The instance of vertex cover problem used in our reduction; red vertices are the vertex cover}
 \label{fig:3-vertex-cover-example}}
\end{minipage}
\hfil
\begin{minipage}{0.38\linewidth}
\begin{minipage}[t]{0.22\linewidth}\centering
{\small
  \textbf{R}\\[1mm]
\begin{tabular}{|ccccc|} \hline
\rowcolor{grey} $A$ & $Z$ & $E_1$ & $E_2$ & $E_3$ \\ \hline
$v_1$ & $z$ & $e_1$ & $e_6$ &  $*$ \\
$v_2$ & $z$ & $e_1$ & $e_2$ & $e_7$ \\
$v_3$ & $z$ & $e_2$ & $e_3$ & $*$ \\
$v_4$ & $z$ & $e_4$ & $e_6$ & $e_7$ \\
$v_5$ & $z$ & $e_3$ & $e_4$ & $e_5$ \\
$v_6$ & $z$ & $e_5$ & $*$ & $*$\\
\hline
\end{tabular}
}
\label{tab:table-3-vertex-cover-table-r-pj}

\end{minipage}
\hfill
\begin{minipage}[t]{0.15\linewidth} \centering
{\small
\textbf{$S_1$}\\[1mm]
\begin{tabular}{|cc|} \hline
\rowcolor{grey} $E$ & $Z$ \\ \hline
$e_1$ & $z$ \\
\hline
\end{tabular}
}
...
{\small
\textbf{$S_7$}\\[1mm]
\begin{tabular}{|cc|} \hline
\rowcolor{grey} $E$ & $Z$ \\ \hline
$e_7$ & $z$ \\
\hline
\end{tabular}
}

 \label{tab:table-3-vertex-cover-table-s-pj}
\end{minipage}
\end{minipage}
\end{minipage}

\cut{
\centering
\begin{minipage}{0.88\linewidth} \centering
{\small\upshape
\begin{tabular}{|c|c|c|} \hline
Queries &  & Result over $R, S$ \\ \hline
$q_1$ & $\projection_{Z}(R \join_{R.E_1 = S_1.E \lor R.E_2 = S_1.E \lor R.E_3 = S_1.E} S_1)$ & $(z)$\\ \hline
$q_2$ & $\projection_{Z}(R \join_{R.E_1 = S_2.E \lor R.E_2 = S_2.E \lor R.E_3 = S_2.E} S_2)$ & $(z)$\\ \hline
... & & \\ \hline
$q_7$ & $\projection_{Z}(R \join_{R.E_1 = S_7.E \lor R.E_2 = S_7.E \lor R.E_3 = S_7.E} S_7)$ & $(z)$\\ \hline
$Q_1$ & $q_1(R,S_1) \join q_2(R, S_2) ... \join Q_7(R, S_7)$ & (z)\\
\hline
\end{tabular}
}
\label{tab:table-3-vertex-cover-reduction-queries-pj}
\end{minipage}
\caption{Example Reduction in Theorem~\ref{theo:swp-is-hard-pj}}
\label{example:encode-3-vertex-cover-pj}
}
\end{figure*}

}

\begin{proof}
We prove the theorem by a reduction from the vertex cover problem with vertex degree at most 3, which is known to be \NPcomplete{} \cite{garey1976some} and is defined as follows:
  Given an undirected graph $G(V, E)$ with vertex set $V$ and edge set $E$, where the degree of every vertex is at most 3, decide whether there exists a vertex cover $C$ of at most $p$ vertices such that each edge in $E$ is adjacent to at least one vertex in the set. 
\par
\textbf{Construction.~} Given $G(V, E)$, suppose $V = \{v_1, ..., v_n\}$, and $E = \{e_1, \cdots, e_m\}$.
We encode each vertex as a tuple in the relation $R(A, Z, E_1, E_2, E_3)$. For each vertex $v_i \in V$, $R$ contains a tuple $t_i = (v_i, z, e_{i1}, e_{i2}, e_{i3})$, where $e_{i_1}, e_{i_2}, e_{i_3}$ are identifiers of edges adjacent to $v_i$, $i_1 < i_2 < i_3$. If the degree of $v_i$ is less than 3, the identifiers are replaced by a null symbol ``$*$''. The attribute $Z = z$ is a constant for all tuples. In addition to $R$, we have $m$ relations $S_1,...,S_m$. Each $S_i$, $i \in [1, m]$, has schema $S_i(E, Z)$. For the edge $e_i \in E$,  $S_i$ contains a single tuple $(e_i, z)$. Let $D = (R, S_1, ..., S_m)$ be the database instance.

  Next, we construct $Q_1$ involving $PJ$ that consist of $m$ subqueries as follows:
  For all $i \in [1, m]$, let $q_i = $ \\ $\projection_{Z}(R $ $\join_{R.E_1 = S_i.E \lor R.E_2 = S_i.E \lor R.E_3 = S_i.E} S_i)$, which operates on $S_i$ and $R$, checks for match of $R.E_1, R.E_2$, or $R.E_3$ with  $S_i.E$, and then projects on to $Z$.  Then we construct $Q_1$ = $q_1 \join q_2 \join ... \join q_m]$ using natural join on $Z$. All queries $q_i$ and $Q_1$ have a single attribute $Z$. Note that, initially, $q_i(D) = \{(z)\}$ for all $i \in [1, m]$, and therefore $Q_1(D) = \{(z)\}$ as well. The query $Q_2$ also outputs the attribute $Z$, but not the tuple $\{(z)\}$. We set $Q_2 = \projection_{Z}(R \join_{R.Z \neq S_1.Z} S_1)$ (the choice of $S_1$ is arbitrary), and therefore $Q_2(D) = \{\}$ is empty.
  The tuple $t$ for which we want to find the smallest witness in $(Q_1 - Q_2)(D)$ is $(z)$. In other words, the goal is to find a subinstance $D' = (R', S'_1,...,S'_m),$ $R' \subseteq R, S'_1 \subseteq S_1, ..., S'_m \subseteq S_m$, such that 
  $(z) \in Q_1(D) \setminus Q_2(D)$. 
\par  
  Below we argue that \emph{$G$ has a vertex cover of size $\leq p$, if and only if the SWP instance above has a witness $D'$ of size $\leq p+m$ where $m$ is the number of edges  in $G$}.

\par
\textbf{The ``Only If'' direction.~} Suppose we are given a vertex cover $C$ with at most $p$ vertices in $G$. 
We construct $R_i' = \{t_j ~|~ v_j \in C\}$, and $S_i' = S_i$ for all $i \in [1, m]$.
Since $|C| \leq p$, $|D'| \leq p+m$ since each $S_i$ contains a single tuple.
Since $C$ is a vertex cover, for all edge $e_i = (v_j, v_\ell) \in E$, either $v_j \in C$ or $v_\ell \in C$. Suppose without loss of generality $v_j \in C$. Then  (wlog.) assume $t_j = (v_j, z, e_i, e', e'')$ where $e', e''$ are other two adjacent edges on $v_j$ (they could be $*$ as well if the degree of $v_j$ is $<3$). The tuple $t_j$ and the tuple $S_i(e_i, z)$ will satisfy the join condition of $q_i$ (irrespective of the position of $e_i$ in $t_i$), and the projection will output $(z)$. Since $C$ is a vertex cover, for all $i \in [1, m]$, $q_i(D') = \{(z)\}$. Therefore, $Q_1(D') = \{(z)\}$. $Q_2(D')$ remains empty. Hence $(z) \in Q_1(D') \setminus Q_2(D')$ Therefore, $D'$ is a witness of $(z)$ of size at most $p+m$. 

\par
\textbf{The ``If'' direction.~} For the opposite direction, consider a witness $D' = (R', S'_1, ..., S'_m)$ where $R' \subseteq R, S'_1 \subseteq S_1, ..., S'_m \subseteq S_m, |R'|+|S'_1|+...+|S'_m| \leq p+m$, such that $(z) \in Q_1(D') \setminus Q_2(D')$, \ie, $(z) \in Q_1(D')$.  We construct $C = \{v_i ~|~ t_i \in R'\}$. 
Note that if $(z) \in Q_1(D')$, $(z)$ must be in the result of all subqueries $q_i(D')$, $i \in [1, m]$. And $q_i(D')$ returns $(z)$ if and only if (a) $S'_i$ is not empty (\ie, $S_i' = S_i$ since $S_i$ had only one tuple), and (b) if $e_i = (v_j, v_\ell)$, at least one of $t_j$ or $t_{\ell}$ must appear in $R'$ to satisfy the join condition in $q_i$; otherwise $q_i$ returns an empty result and thus $Q_1$ returns an empty result. Therefore, all $S'_i$ must be equal to $S_i$, $|S'_i| = 1$. Then we have $|S'_1|+...+|S'_m| = m$. Since $|D'| \leq p+m$, $|R'| \leq p$,  and thus we get a vertex cover $C$ of size at most $p$.

  An example reduction is shown in Figure~\ref{fig:Thm:PJ}.
\end{proof}

\begin{figure}[t]
\scriptsize
\begin{minipage}{0.92\linewidth}\centering
\begin{minipage}{0.43\linewidth}\centering
\begin{minipage}{1.0\linewidth}\centering
\subfigure[$G(V, E)$\label{fig:ThmPJ-VC}]{
 \centering
      \resizebox{\linewidth}{!}{
\begin{tikzpicture}
  [scale=.8,auto=left,every node/.style={circle,fill=blue!20,scale=1.5}]
  \node (n6) at (10,10) {v6};
  \node (n4) at (3,10)  {v4};
  \node[fill=red] (n5) at (7,10)  {v5};
  \node[fill=red] (n1) at (1,8) {v1};
  \node[fill=red] (n2) at (3,6)  {v2};
  \node (n3) at (7,6)  {v3};

  \foreach \from/\to/\label in {n2/n1/e1,n2/n3/e2,n5/n3/e3,n4/n5/e4,n5/n6/e5,n1/n4/e6,n4/n2/e7}
    \draw (\from) -- (\to) node [midway, fill=white] {\large \label};

\end{tikzpicture}
}}
\end{minipage}
\vfil
\begin{minipage}[t]{1.0\linewidth}\centering
\subfigure[$S_1, \cdots, S_7$\label{fig:ThmPJ-S}]{
{\small
\begin{tabular}{|cc|} \hline
\rowcolor{grey} $E$ & $Z$ \\ \hline
$e_1$ & $z$ \\
\hline
\end{tabular}
}
...
{\small
\begin{tabular}{|cc|} \hline
\rowcolor{grey} $E$ & $Z$ \\ \hline
$e_7$ & $z$ \\
\hline
\end{tabular}
}
}
\end{minipage}
\end{minipage}
\hfill
\begin{minipage}{0.43\linewidth}
\begin{minipage}[t]{0.22\linewidth}\centering
\subfigure[R\label{fig:ThmPJ-R}]{
{\small
\begin{tabular}{|ccccc|} \hline
\rowcolor{grey} $A$ & $Z$ & $E_1$ & $E_2$ & $E_3$ \\ \hline
$v_1$ & $z$ & $e_1$ & $e_6$ &  $*$ \\
$v_2$ & $z$ & $e_1$ & $e_2$ & $e_7$ \\
$v_3$ & $z$ & $e_2$ & $e_3$ & $*$ \\
$v_4$ & $z$ & $e_4$ & $e_6$ & $e_7$ \\
$v_5$ & $z$ & $e_3$ & $e_4$ & $e_5$ \\
$v_6$ & $z$ & $e_5$ & $*$ & $*$\\
\hline
\end{tabular}
}
}
\end{minipage}
\end{minipage}
\end{minipage}
\caption{\label{fig:Thm:PJ} Example reduction in Theorem~\ref{thm:PJ-hard}}
\end{figure}

\vspace{-3mm}
\subsection{JU Queries}\label{sec:JU}


\begin{Theorem}
\label{thm:JU-hard}
The SWP for two JU queries is \NPhard{} in query complexity.
\end{Theorem}

\begin{proof}
We reduce from the vertex cover problem.
\par
\textbf{Construction.~} 
Suppose $V = \{v_1, ..., v_n\}$ and $E = \{e_1,\cdots,$ \\$e_m\}$.
For each vertex $v_i$ in $G$, there is a relation $R_i(Z)$ which consists of a single tuple $(z)$. 
For each edge $e_i = (v_j, v_\ell) \in E$, we construct a query $q_i = R_j \union R_\ell$. Then we construct a query $Q_1 = q_1 \join \cdots \join q_m$, where the join is a natural join on $Z$. We construct $Q_2 = R_1 \join_{R_1.Z \neq R_2.Z} R_2$ (the choice of $R_1, R_2$ is arbitrary). Hence $D = (R_1, \cdots, R_n)$, $Q_1(D) = \{(z)\}$, and $Q_2(D) = \{\}$. The output tuple $(z) \in Q_1(D) \setminus Q_2(D)$, and the goal is to find a witness $D' =$ $(R_1', \cdots,$\\ $R_n')$ for $(z)$ where $R_i' \subseteq R_i$ for all $i \in [1, n]$. 

We show that \emph{there exists a vertex cover $C$ in $G$ of size $\leq p$ if and only if there is a witness $D'$ for $(z)$ of size $\leq p$}. 
\par
\textbf{The ``Only If'' direction.~}   Consider a vertex cover $C$ of $G$ such that $|C| \leq p$. 
If $v_i \in C$, then $R_i' = \{(z)\}$, otherwise $R'_i = \{\}$. Since $C$ is a vertex cover, all edges must be covered.  For an edge $e_i = (v_j, v_\ell)$, suppose wlog. $v_j \in C$. Hence the subquery $q_i= R_{j} \union R_{\ell}$ returns $(z)$ on $D'$. Therefore, $Q_1(D') = (z), Q_2(D') = \{\}$, $(z) \in Q_1(D') \setminus Q_2(D')$, \ie, $D'$ is a witness for $(z)$, and $|D'| = |C| \leq p$.
\par
\textbf{The ``If'' direction.~} Consider any witness $D' = (R'_1, ..., R'_n)$ where $R'_1 \subseteq R_1, ..., R'_n \subseteq R_n$ and $|R'_1|+...+|R'_n| \leq p$, such that $(z) \in Q_1(D') \setminus Q_2(D')$, \ie, $(z) \in Q_1(D')$. Since $R_i$ had only one tuple $(z)$, either $R'_i$ has $(z)$ or it is empty. If tuple $(z) \in R'_i$, then we add vertex $v_i$ to a set $C$. If $(z)$ is in the result of $Q_1(D')$, $(z)$ must be in the result of all subqueries $q_i(D')$ for all $i \in [1, m]$. For $e_i = (v_j, v_{\ell})$,  $q_i(D')$ returns $(z)$ if and only if at least one of $R'_{j}$ and $R'_{\ell}$ is not empty; otherwise $q_i$ returns an empty result and thus $Q_1$ returns an empty result. Therefore, for each edge $e_i \in E$, at least one of its adjacent vertices $v_j$ or $v_\ell$ must exist in $C$.
Hence $C$ is a vertex cover, and $|C| = |D'| \leq p$.
\end{proof}

On the other hand, the following theorem shows that if all unions appear after all joins (which we call JU$^*$ queries), then the SWP can be solved in poly-time in combined complexity.

\begin{Theorem}\label{thm:JU-restricted-poly}
The SWP for two JU$^*$ queries is polynomial time solvable in combined complexity.
\end{Theorem}
\begin{proof}
Let $t \in Q_1(D) \setminus Q_2(D)$. 
According to Theorem~\ref{thm:sj-poly}, the SWP for SJ queries is polynomial time solvable in combined complexity. 
Hence, we look for the smallest witness of $t$ in join-only part of $Q_1$, and choose the one with smallest number of tuples. The running time is polynomial in both $n = |D|$ and $k$.
\end{proof}

\vspace{-3mm}
\subsection{Size-Bounded SPJU Queries}\label{sec:SPJU}
Theorem~\ref{thm:SPJU-poly} shows that if the SPJU queries are of bounded size (i.e. considering more standard data complexity), there is a polynomial time algorithm for SWP. We prove this theorem using Proposition~\ref{prop:prov-dnf}, which is intuitive and known (e.g., \cite{DBLP:conf/icdt/RoyPT11}). We use \emph{$m$-DNF} to refer to a DNF where each minterm has at most $m$ literals. 

\begin{proposition}
\label{prop:prov-dnf}
Given an SPJU query $Q$, a database instance $D$, and an output tuple $t \in Q(D)$, the how-provenance of $t$ in $Q(D)$ can be transformed into a $k+1$-DNF in polynomial time when $Q$ is of bounded size, where $k$ is the  number of join operations in $Q$.
\end{proposition}

\cut{
\begin{proof}
\red{(this proof has not been edited and can be omitted if we run out of space.)}
We prove this lemma by induction on $h$, where $h$ is the maximum number of join operations in any path from the root to leaves in the query tree (logical query plan) of $Q$. When $h=0$, there is no relational operator in $Q$, $Q$ returns the input relation. Then the how-provenance of an output tuple $t$ is the identifier of $t$ itself, which is a $1$-DNF of a single clause. Therefore the induction hypothesis holds when $h=0$.

Suppose the induction hypothesis holds for all SPJU queries with query tree of height $ \leq h$, then consider an SPJU query $Q$ whose query tree of height $h+1$.
\cut{We prove the hypothesis by the cases of the relational operator $op$ at the root of $Q$'s query tree.}

\begin{itemize}
\item $op$ is selection, $Q = \sigma_{\theta}(Q')$. The how-provenance of a tuple $t \in Q(D)$ $=\lambda(t)$ remains the same as $\lambda'(t')$ for $t' = t$ in $Q'(D)$ if $t'$ satisfies the selection condition; otherwise $t$ is not in the output of $Q$.
\item $op$ is projection, $Q = \projection_{U}(Q')$. $\lambda(t)$ is the disjunction of how-provenance of tuples in the result of $Q'(D)$ for all tuples $t'$ where $t'[U] = t$. The result of $Q'(D)$ contains at most $m^q$ tuples, where $m$ is the maximum size of input relations in $D$, $q$ is the size of $Q'$. The size of clauses in $\lambda(t)$ in DNF form is the same as the size of clauses in $\lambda'(t')$ for $t' \in Q'(D)$.
\item $op$ is union, $Q = Q'_1 \union Q'_2$. $\lambda(t)$ remains the same as $\lambda_1(t_1)$ for $t_1 \in Q'_1(D)$ or $\lambda_2(t_2)$ for $t_2 \in Q'_2(D)$.
\item $op$ is join, $Q = Q'_1 \join Q'_2$. $\lambda(t)$ is the conjunction of  $\lambda_1(t_1)$ and $\lambda_2(t_2)$ for $t_1 \in Q'_1(D)$, $t_2 \in Q'_2(D)$, and $t_1 = t[U_1], t_2 = t[U_2]$, where $U_1$ and $U_2$ is the schema of $Q'_1(D)$ and $Q'_2(D)$ respectively. Assume that $Q'_1$ has $r$ join operators in its query tree, and $Q'_2$ has $s$ join operators in its query tree, then $lambda_1(t_1)$ is a $(r+1)$-DNF, $lambda_2(t_2)$ is a $(s+1)$-DNF. Expanding the conjunction to obtain a DNF is $O(m^q_1 \times m^q_2) = O(m^q)$, where $m$ is the maximum size of input relations in $D$, $q_1, q_2, q$ are the sizes of $Q'_1, Q'_2, Q'$. Then $lambda(t)$ is a $(r+s+2)$-DNF, and the query tree of $Q$ has $r+s+1$ join operators, therefore the inductive hypothesis holds.
\end{itemize}
\end{proof}
}
\begin{Theorem}
\label{thm:SPJU-poly}
The SWP for two SPJU queries is polynomial-time solvable in data complexity.
\end{Theorem}

\begin{proof}
Let $t$ be an output tuple in $Q_1(D) \setminus Q_2(D)$. Since $Q_2$ is monotone, $t \notin Q_2(D')$ for any $D' \subseteq D$. According to Proposition~\ref{prop:prov-dnf}, we can compute the how-provenance $\Prv_{(Q_1-Q_2)(D)}(t)$ in DNF in poly-time in data complexity. Then we scan the DNF to find the minterm with least number of literals, and this minterm represents the smallest witness for $t$ in $Q_1(D) - Q_2(D)$. The literals in this clause are the identifiers of tuples in the smallest witness.
\end{proof}
For instance, if $\Prv(t) = a + bc$, then $a$ forms the smallest witness.

\vspace{-3mm}
\subsection{Queries Involving Difference}


Before discussing general SPJUD queries, let's focus on one special class of SPJUD queries where all differences appear after all SPJU operators (which we call SPJUD$^*$ queries). More formally, we define this class using formal grammar: $Q \rightarrow q^+ | Q - Q$, where $q^+$ is a terminal that represents SPJU queries. For instance, queries $Q_1$ and $Q_2$ in Example~\ref{example:intro1} are SPJUD$^*$ queries. The following theorem shows that the SWP can be solved in poly-time for SPJUD$^*$ queries.
\begin{Theorem}\label{thm:SPJUD-restricted-poly}
The SWP for two SPJUD$^*$ queries is polynomial-time solvable in data complexity.
\end{Theorem}

\begin{proof}
Let $t$ be an output tuple in $Q_1(D) \setminus Q_2(D)$. Since $Q_1$ and $Q_2$ are SPJUD$^*$ queries that can be written as nested differences of queries like $q_1 - q_2 - (q_3 - (q_4 - q_5)) - ...$, where all $q_i$-s are SPJU queries, $Q_1 - Q_2$ is also an SPJUD$^*$ query. The output tuple $t$ must be either in or not in the result of each $q_i$. We find the smallest witness by enumerating the minimal witnesses of $t$ w.r.t. every $q_i$ and $D$. If $t$ is in the result of $q_i(D)$, let $w_i$ be the set of minimal witnesses of $t$ w.r.t. $q_i$ and $D$. Then we pick one element from every $w_i \union \{\emptyset\}$, and construct $w$ as the union of all witnesses or the empty set we picked. We evaluate $Q_1$ and $Q_2$ on $w$ to see whether it is a witness for $t$, and record the $w$ of the smallest size. We finish this procedure until we enumerate all combinations.



This procedure will return the smallest witness because: (i) if $t \notin q_i(D)$, $t$ will also not be in $q_i(w)$ for any $w \subseteq D$ due to monotonicity, so we don't need to consider such $q_i$-s; (ii) Assume that $w'$ is a smallest witness of $t$ w.r.t. $Q_1-Q_2$ and $D$, for all $q_i$ where $t \in q_i(w')$, $w'$ must be a superset of a minimal witness of $t$ w.r.t. $q_i$ and $D$. Hence $w'$ must be the union of minimal witnesses of $t$ w.r.t. these $q_i$-s and $D$; otherwise, if $w'$ is a strict superset of the union of minimal witnesses of $t$, we can always remove tuples not belong to any minimal witness of $t$ w.r.t. $q_i$-s and $D$ from $w'$, without affecting $t$ to be in or not in any $q_i$, which contradicts the assumption that $w'$ is a smallest witness. Therefore a smallest witness of $t$ w.r.t. $Q_1 - Q_2$ and $D$ must be union of minimal witness of $t$ w.r.t. $q_i$ and $D$, and thus it must be enumerated during the enumeration procedure.

The time complexity of entire enumeration process is $O(\Pi_{i} m^{k_i}) = O(m^{kd})$, where d is the number of difference operators, m is the max size of relations, k is the max complexity of each SPJU query $q_i$. When queries are of bounded sizes, i.e., fix d and k, the SWP for two SPJUD queries that can be written as nested differences of SPJU queries is polynomial-time solvable. 
\end{proof}



SWP is NP-hard in general even for bounded-size queries.
\begin{Theorem}\label{thm:SPJUD-hard}
The SWP for two SPJUD queries $Q_1$ and $Q_2$ is \NPhard{} in data complexity.
\end{Theorem}

\begin{proof}
  We again give a reduction from the vertex cover problem with vertex degree at most 3 (see Theorem~\ref{thm:PJ-hard}).
\par
\textbf{Construction.~} 
  Suppose in $G = (V, E)$, $V = \{v_1, ..., v_n\}$, $E = \{e_1, \cdots, e_m\}$. We construct two relations $R(A, Z, E_1, E_2, E_3)$ and $S(B, C, Z)$. For each vertex $v_i \in V$, $R$ contains a tuple $t_i  = (v_i, e_{i1}, e_{i2}, e_{i3}, z)$, where $e_{i_1}, e_{i_2}, e_{i_3}$ are the identifiers of edges adjacent to $v_i$, $i_1 < i_2 < i_3$. If the degree of $v_i$ is less than 3, the identifiers are replaced by a null symbol ``$*$''. Here $z$ is a constant. For each edge $e_i \in E$, $S$ contains a tuple $(e_i, e_{(i \% m) + 1}, z, z)$, where $e_{(i \% m) + 1}$ is the identifier of the next edge in the edge list (the next edge of $e_m$ is $e_1$). Let $D = (R, S)$ be the database instance.

  Next, we construct an SPJUD query that consists of several subqueries as follows:
  Let $q_1$ (on $S$) = $\projection_{Z}(S)$; 
  $q_2$ (on $S$) $=\projection_{B, Z}(S)$; 
  $q_3$ (on $R, S$)=$\projection_{S.C, S.Z}(S \join_{S.C = E_1 \lor S.C = E_2 \lor S.C = E_3} R)$.
  Then we construct $Q_1 = q_1$, hence $Q_1(D) = \{(z)\}$. We also construct $Q_2 = \projection_{Z}(q_2 \setminus q_3)$ (assume $C$ in $q_3$ is renamed to $B$). For edge $e_i = (v_j, v_{\ell})$, the edge $e_i$ appears for both  tuples $t_j, t_\ell$ (in one of $E_1, E_2, E_3$ attributes), and therefore, $(e_i, z)$ appears in the result of $q_3(D)$ for every $i \in [1, m]$. Hence $q_3(D) = \projection_{B, Z}(S)$. So $q_2(D) \setminus q_3(D) = \emptyset$. Then $(Q_1-Q_2)(D) = \{(z)\}$, and the goal is to find the smallest witness for $(z)$.
For the vertex cover instance in Figure~\ref{fig:ThmPJ-VC}, $R$ will be as given in Figures~\ref{fig:ThmPJ-R}, and $S$ will contain tuples $\{(e_1, e_2, z), (e_2, e_3, z), \cdots (e_7, e_1, z)\}$. 

\cut{
\begin{figure}[t]
{\scriptsize
\begin{tabular}{|c|c|c|} \hline
Queries &  & Result over $R, S$ \\ \hline
$q_1$ & $\projection_{Z}(S)$ & $(z)$ \\ \hline
$q_2$ & $\projection_{B, Z}(S)$ & $(e_1, z), (e_2, z)$, \\
& & $\cdots, (e_7, z)$ \\ \hline
$q_3$ & $\projection_{S.C, S.Z}(S \join_{\phi} R)$ & $(e_2, z), ..., $\\ 
& $\phi: (S.C = E_1 \lor S.C = E_2 \lor S.C = E_3)$ & $(e_7, z), (e_1, z)$\\\hline
$Q_1$ & $q_1$& $(z)$\\ \hline
$Q_2$ & $\projection_{Z}(q_2 \setminus q_3)$ & $\emptyset$\\
\hline
\end{tabular}
}

\caption{Queries in the reduction in Theorem~\ref{thm:SPJUD-hard}}
\label{fig:SPJUD}
\end{figure}
}

  We now show that \emph{there exists  a vertex cover $C$ of size at most $p$ in the graph $G$ if and only if there is a witness $D' = (R', S')$ where $|R'| + |S'| \leq p+m$}.
  \par
\textbf{The ``Only If'' direction.~}   
  Suppose we are given a vertex cover $C$ of $G$ with at most $k$ vertices. Construct $R' \{t_i ~|~ v_i \in C\}$, and $S' = S$. $Q_1(D) = Q_1(D') = \{(z)\}$ since $S$ is unchanged. Similarly, $q_2(D') = \projection_{B, Z}(S))$ is unchanged. Since $C$ is a vertex cover, for every edge $e_i = (v_j, v_{\ell})$ either $t_i$ or $t_{\ell}$ is in $R'$,  and hence $q_3(D') = q_3(D)$, \ie, each $(e_i, z)$, $i \in [1, m]$ appears in $q_3(D')$. Hence $Q_1 - Q_2$ will output tuple $(z)$ on $D'$, $|R'|  = |C| \leq p, |S'| = |S| = m$, and we get a  witness of at most $p + m$ tuples.
\par
  \textbf{The ``If'' direction.~} Consider any witness $D' = (R', S')$ where $R' \subseteq R, S' \subseteq S, |R'|+|S'| \leq p+m$, such that $(z) \in Q_1(D') \setminus Q_2(D')$. We construct $C = \{v_i ~|~ t_i \in R'\}$. Since $(z)$ is in $Q_1(D') \setminus Q_2(D')$, $(z)$ must be in the result of $q_1(S')$, and not in the result of $q_2(S') - q_3(R', S')$, hence $S'$ must contain at least one tuple. Therefore, $q_2(S')$ outputs at least one tuple $(e_i, z)$ since $S'$ is not empty. In turn, $q_3(R', S')$ must output all tuples in $q_2(S')$ to make $q_2(S') - q_3(R',S')$ empty. 
(a) We argue that $S'  = S$. 
Suppose $S'$ contains at least one tuple, say wlog, $(e_1, e_2, z)$. Then to remove $(e_1, z)$ from $q_2(S') \setminus q_3(R', S')$, $q_3(R', S')$ must contain $(e_1, z)$, which can generate only from $S(e_m, e_1, z)$. Hence $(e_m, e_1, z) \in S'$. In turn, $(e_m, z) \in q_2(S')$. To remove it, we need $S(e_{m-1}, e_m, z)$ in $S'$. Continuing this argument (by induction), we get $S = S'$. 
(b) Consider any tuple, say wlog.,  $(e_1, e_2, z)$ in $S'$. Then to remove $(e_1, z)$ from $q_2(S') \setminus q_3(R', S')$, 
not only the tuple $(e_m, e_1, z) \in S'$, it also has to satisfy the join condition with $R$. This will hold only if for one of the end points $v_j, v_\ell$ of $e_1 = (v_j, v_\ell)$, $t_j \in R'$ or $t_{\ell} \in R'$. This should hold for all edges, and therefore the set $C$ we constructed is a vertex cover. Since $|S'| = |S| = m$, $|R'| = |C| \leq p$, therefore, we get a vertex cover in $G$ of size at most $p$.

 The queries we constructed during the reduction are all of bounded size, therefore the SWP for two  SPJUD queries is NP-hard in data complexity even for queries of bounded size.
\end{proof}

\cut{

\begin{figure*}[!h]
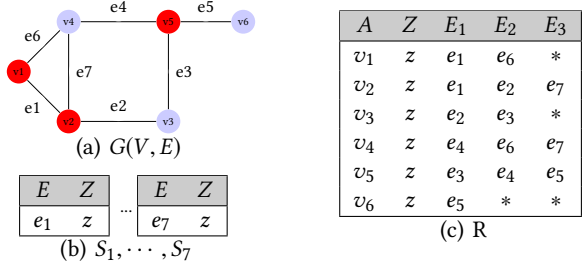



\begin{minipage}{0.98\linewidth}\centering
\begin{minipage}{0.30\linewidth}\centering
\begin{minipage}[t]{0.05\linewidth}\centering
{\small
  \textbf{$R_1$}\\[1mm]
\begin{tabular}{|c|} \hline
\rowcolor{grey} $Z$ \\ \hline
$z$\\
\hline
\end{tabular}
}
\end{minipage}
\hfil
\begin{minipage}[t]{0.05\linewidth}\centering
{\small
  \textbf{$R_2$}\\[1mm]
\begin{tabular}{|c|} \hline
\rowcolor{grey} $Z$ \\ \hline
$z$\\
\hline
\end{tabular}
}
\end{minipage}
\hfil
\begin{minipage}[t]{0.05\linewidth}\centering
{\small
  \textbf{$R_3$}\\[1mm]
\begin{tabular}{|c|} \hline
\rowcolor{grey} $Z$ \\ \hline
$z$\\
\hline
\end{tabular}
}
\end{minipage}
\hfil
\begin{minipage}[t]{0.05\linewidth}\centering
{\small
  \textbf{$R_4$}\\[1mm]
\begin{tabular}{|c|} \hline
\rowcolor{grey} $Z$ \\ \hline
$z$\\
\hline
\end{tabular}
}
\end{minipage}
\hfil
\begin{minipage}[t]{0.05\linewidth}\centering
{\small
  \textbf{$R_5$}\\[1mm]
\begin{tabular}{|c|} \hline
\rowcolor{grey} $Z$ \\ \hline
$z$\\
\hline
\end{tabular}
}
\end{minipage}
\hfil
\begin{minipage}[t]{0.05\linewidth} \centering
{\small
\textbf{$R_6$}\\[1mm]
\begin{tabular}{|c|} \hline
\rowcolor{grey} $Z$ \\ \hline
$z$ \\
\hline
\end{tabular}
}
\end{minipage}
\end{minipage}
\hfill
\begin{minipage}{0.65\linewidth} \centering
{\small\upshape
\begin{tabular}{|c|c|c|} \hline
Queries &  & Result over $R_1, ..., R_n$ \\ \hline
$q_1$ & $R_{1} \union R_{2} $ & $(z)$\\ \hline
$q_2$ & $R_{2} \union R_{3}$  & $(z)$\\ \hline
... & & \\ \hline
$q_7$ & $R_{2} \union R_{4}$ & $(z)$\\ \hline
$Q_1$ & $q_1(R_{1}, R_{2}) \join q_2(R_{2}, R_{3}) ... \join q_7(R_{2}, R_{4})$ & $(z)$\\
\hline
\end{tabular}
}
\label{tab:table-3-vertex-cover-reduction-queries-ju}
\end{minipage}
\caption{Example Reduction in Theorem~\ref{theo:swp-is-hard-ju}}
\label{example:encode-3-vertex-cover-ju}
\end{minipage}
\end{figure*}
}

\cut{
\begin{figure*}[!h]


\centering
\begin{minipage}{0.98\linewidth}
\begin{minipage}{0.3\linewidth}\centering
\begin{minipage}[t]{0.6\linewidth}
{\small
  \textbf{R}\\[1mm]
\begin{tabular}{|ccccc|} \hline
\rowcolor{grey} $A$ & $Z$ & $E_1$ & $E_2$ & $E_3$ \\ \hline
$v_1$ & $z$ & $e_1$ & $e_6$ &  $*$ \\
$v_2$ & $z$ & $e_1$ & $e_2$ & $e_7$ \\
$v_3$ & $z$ & $e_2$ & $e_3$ & $*$ \\
$v_4$ & $z$ & $e_4$ & $e_6$ & $e_7$ \\
$v_5$ & $z$ & $e_3$ & $e_4$ & $e_5$ \\
$v_6$ & $z$ & $e_5$ & $*$ & $*$\\
\hline
\end{tabular}
}
\label{tab:table-3-vertex-cover-table-r-spjud}

\end{minipage}
\hfil
\begin{minipage}[t]{0.3\linewidth}
{\small
\textbf{S}\\[1mm]
\begin{tabular}{|ccc|} \hline
\rowcolor{grey} $B$ & $C$ & $Z$ \\ \hline
$e_1$ & $e_2$ & $z$ \\
$e_2$ & $e_3$ & $z$ \\
$e_3$ & $e_4$ & $z$ \\
$e_4$ & $e_5$ & $z$ \\
$e_5$ & $e_6$ & $z$ \\
$e_6$ & $e_7$ & $z$ \\
$e_7$ & $e_1$ & $z$ \\
\hline
\end{tabular}
}
 \label{tab:table-3-vertex-cover-table-s-spjud}
\end{minipage}
\end{minipage}
\hfill
\begin{minipage}{0.65\linewidth}\centering
{\small\upshape
\begin{tabular}{|c|c|c|} \hline
Queries &  & Result over $R, S$ \\ \hline
$q_1$ & $\projection_{Z}(S)$ & $(z)$ \\ \hline
$q_2$ & $\projection_{B, Z}(S)$ & $(e_1, z), (e_2, z), ..., (e_7, z)$ \\ \hline
$q_3$ & $\projection_{S.C, S.Z}(S \join_{S.C = E_1 \lor S.C = E_2 \lor S.C = E_3} R)$ & $(e_2, z), ..., (e_7, z), (e_1, z)$\\ \hline
$Q_1$ & $q_1(S)$& $(z)$\\ \hline
$Q_2$ & $\projection_{Z}(q_2(S) \setminus q_3(R, S))$ & $\emptyset$\\
\hline
\end{tabular}
}
\label{tab:table-3-vertex-cover-reduction-queries-spjud}
\end{minipage}
\end{minipage}

\caption{Example Reduction in Theorem~\ref{theo:swp-is-hard}}
\label{example:encode-3-vertex-cover-spjud}
\end{figure*}
}


\end{appendix}

\end{document}